\documentclass[english]{amsart}
\usepackage[utf8]{inputenc}
\usepackage{biblatex}
\usepackage{geometry}
\usepackage{babel}
\usepackage{mathtools}
\usepackage{amstext}
\usepackage{amsthm}
\usepackage{amssymb}
\usepackage{stmaryrd}
\usepackage[unicode=true,pdfusetitle,
 bookmarks=true,bookmarksnumbered=false,bookmarksopen=false,
 breaklinks=false,pdfborder={0 0 1},backref=false,colorlinks=false]
 {hyperref}
\makeatletter
\numberwithin{equation}{section}
\numberwithin{figure}{section}
\theoremstyle{plain}
\newtheorem{thm}{\protect\theoremname}[section]
\theoremstyle{remark}
\newtheorem{claim}[thm]{\protect\claimname}
\theoremstyle{plain}
\newtheorem{lem}[thm]{\protect\lemmaname}
\theoremstyle{plain}
\newtheorem{cor}[thm]{\protect\corollaryname}
\theoremstyle{plain}
\newtheorem{prop}[thm]{\protect\propositionname}
\theoremstyle{definition}
\newtheorem{defn}[thm]{\protect\definitionname}
\theoremstyle{remark}
\newtheorem{rem}[thm]{\protect\remarkname}
\theoremstyle{definition}

\theoremstyle{plain}

\usepackage{subcaption}
\usepackage{graphicx}
\usepackage{xcolor}
\usepackage{tikz}
\usepackage{pgfplots}
\usetikzlibrary{positioning, shapes.geometric, arrows}

\tikzstyle{startstop} = [rectangle, rounded corners, 
text centered, 
draw=black, 
fill=red!30]

\tikzstyle{decision} = [rectangle,
text width=2.5cm, 
text centered, 
draw=black]
\tikzstyle{arrow} = [thick,->,>=stealth]

\global\long\def\Z{\mathbb{Z}}%
\global\long\def\N{\mathbb{N}}%
\global\long\def\M{\mathcal{M}}%
\global\long\def\U{\mathbb{U}}%
\global\long\def\diff{\,d}%
\global\long\def\tendsto#1#2{\underset{#1\to#2}{\longrightarrow}}%
\global\long\def\A{\mathcal{A}}%
\global\long\def\G{\mathcal{G}}%
\global\long\def\F{\mathcal{F}}%
\global\long\def\X{\mathcal{X}}%
\global\long\def\define#1{\textit{#1}}
\makeatother

\providecommand{\corollaryname}{Corollary}
\providecommand{\definitionname}{Definition}
\providecommand{\claimname}{Claim}
\providecommand{\examplename}{Example}
\providecommand{\lemmaname}{Lemma}
\providecommand{\propositionname}{Proposition}
\providecommand{\remarkname}{Remark}
\providecommand{\theoremname}{Theorem}
\providecommand{\questionname}{Question}

\usepackage{marvosym,fancybox}

\begin{document}

\title{A perturbed cellular automaton with two phase transitions for the ergodicity}

\author{Hugo MARSAN}
\address{Institut de Mathématiques de Toulouse, Université Paul Sabatier, Toulouse, France}
\email{hugo.marsan@ens-paris-saclay.fr}

\author{Mathieu SABLIK}
\address{Institut de Mathématiques de Toulouse, Université Paul Sabatier, Toulouse, France}
\email{Mathieu.Sablik@math.univ-toulouse.fr}

\author{Ilkka TÖRMÄ}
\address{Department of Mathematics and Statistics, University of Turku, Turku, Finland}
\email{iatorm@utu.fi}
\thanks{Ilkka Törmä was supported by the Academy of Finland grant 359921, and a visiting researcher grant from Paul Sabatier University. M.~Sablik acknowledges the support of the ANR “Difference” project (ANR-20-CE40-0002).}

\begin{abstract}
    The positive rates conjecture states that a one-dimensional probabilistic cellular automaton (PCA) with strictly positive transition rates must be ergodic.
    The conjecture has been refuted by Gács, whose counterexample is a cellular automaton that is non-ergodic under uniform random noise with sufficiently small rate.
    For all known counterexamples, non-ergodicity has been proved under small enough rates.
    Conversely, all cellular automata are ergodic with sufficiently high-rate noise.
    No other types of phase transitions of ergodicity are known, and the behavior of known counterexamples under intermediate noise rates is unknown.
    
    We present an example of a cellular automaton with two phase transitions.
    Using Gács's result as a black box, we construct a cellular automaton that is ergodic under small noise rates, non-ergodic for slightly higher rates, and again ergodic for rates close to 1.
\end{abstract}

\maketitle

\tableofcontents{}

\section{Introduction}

A configuration is a coloring with elements of $\A$, a finite alphabet, of the sites of the lattice $\Z$. A cellular automaton (CA) is a dynamical system on such configurations, obtained by iterating a local update rule simultaneously at each site of the lattice. This simple model has a wide variety of behaviors and it is used to represent locally interacting phenomena. It is natural to study their perturbed counterpart, Probabilistic Cellular Automata (PCA),  to understand the robustness of their computation to noise. One natural way to define a perturbation is to state that after each iteration of a given cellular automaton, each cell is independently modified with probability~$\epsilon$. In this article the new value is uniformly chosen over the alphabet $\A$. Thus this noise has positive rates: every symbol has a positive probability of appearing regardless of the previous state.

The most natural question about the asymptotic behavior of a probabilistic cellular automaton is its ergodicity. A probabilistic cellular automaton is said to be ergodic if  its action on probability measures has a unique fixed point that attracts all the other measures. This means that it asymptotically “forgets” its initial condition since the distribution of the initial configuration always converges to the same distribution. 

Most cellular automata seem to be ergodic under a positive rates perturbation, and a lot of methods have been developed to prove the ergodicity of probabilistic cellular automata~\cite{Vasilyev-78,Gray-1982,Gacs-Torma-2022}. In \cite{MST19} the authors show that large classes of perturbed CA are ergodic. Constructing a cellular automaton robust to noise in the sense that its trajectories remain distinguishable under the influence of noise is a notoriously difficult problem.

The first examples of robust CA, given by A. Toom~\cite{Toom80}, were two-dimensional; the most famous one is the majority vote with neighborhood $\{(0,0), (0,1), (1,0\}$. In dimension one the \define{positive rates conjecture}, i.e. the conjecture that all positive rates PCA were ergodic, stood on several arguments, including the fact that under various assumptions, one-dimensional Ising models are ergodic, while some simple higher-dimensional Ising models admit phase transitions \cite[Section~IV.3]{Liggett05}. 

In 1986, the positive rates conjecture was refuted by Peter G\'acs in \cite{Gacs86} with an extremely intricate construction of a CA that involves an infinite hierarchy of simulations between increasingly complex generalized cellular automata.
G\'acs later published the extended article \cite{Gacs01} with more detailed proofs and additional results in continuous time (see \cite{Gray01} for a simplified overview), and most recently a further expanded preprint \cite{Gacs24}.
See \cite{Gacs-2024} for a more detailed history of the problem.

In all known examples of CA robust to noise, there is critical value $\epsilon_c$ such that the perturbed cellular automaton is ergodic for any $\epsilon<\epsilon_c$. For $\epsilon\geq\epsilon_c$ the nature of the probabilistic cellular automaton is not known except for $\epsilon$ close to $1$. Indeed, thanks to a percolation argument, it is shown in~\cite{MST19} that given any cellular automaton, there is a constant $\epsilon'$ (which depends only of the radius of the cellular automata) such that the probabilistic cellular automaton is ergodic for any $\epsilon\geq\epsilon'$. Thus, at high noise levels, these cellular automata cannot distinguish trajectories, but if the noise is low enough, these cellular automata are robust to errors. Therefore these examples have at least one transition phase considering ergodicity when the noise varies. 

Those examples raise a question: is it possible for a perturbed cellular automaton to admit several phase transitions? In particular, is there a perturbed cellular automaton which is not ergodic for some $\epsilon$ but which becomes ergodic in the low-noise regime? This seems counter-intuitive since when the noise is small it is easier to keep the trajectories distinguishable. The contribution of this article is to answer these questions: we construct a cellular automaton such that its perturbed version performs at least two phase transitions when the noise varies. It is ergodic when the noise is close to 1 (as any perturbed cellular automata), not ergodic at a fixed $\epsilon\in(0,1)$, and ergodic again when the noise is close to 0 (Theorem~\ref{thm:TransitionDePhase}).  

The idea of the construction is to use the perturbed cellular automaton of Gács which is not ergodic for a small enough noise (it admits a critical value $\epsilon_c$ such that it is not ergodic for $\epsilon\leq\epsilon_c$) and add a layer which controls the quantity of mistakes appearing in the first layer. The second layer is composed by cells with a value between $0$ and $k-1$ which is incremented by $1$ modulo $k$ at each time step, with the value $0$ producing a mistake on the first layer. There are also arrows which move at speed one toward the right, synchronizing the counters they pass over. Those signals cannot live for more than $ak$ time steps. This cellular automaton is described in Section~\ref{section-Description}. 

In Section~\ref{sec:non-ergodicity} we show that if $k$ is large enough and $\epsilon = \frac{1}{\ln\ln\ln(k)}$, for this rate of noise the errors produced by the second layer on the first one are comparable to uniform noise under the critical value $\epsilon_c$, the Gács CA can correct the mistakes. Thus our CA is not ergodic.

On the other hand, when $\epsilon$ goes to $0$, the arrows produced by noise have a high probability to survive a long time before disappearing  and thus synchronize large parts of the spacetime diagram. Such synchronized regions block all information flow on the Gács layer, as the force all cells to make an error simultaneously every $k$ time steps. In Section~\ref{sec:low-noise-ergodicity} we study the dependence cone of the perturbed CA where the borders can be associated to Markov additive chains~\cite{Asmussen03}. If $a$ is large enough, the arrows synchronize sufficiently large zones that the left and the right border collide almost surely.
Hence the initial configuration is forgotten, that is to say the perturbed CA is ergodic.

\section{Perturbations of cellular automata}

\subsection{Definitions}

Let $\A$ be a finite alphabet of symbols. Define $\X = \A^{\Z}$ to be the space of \define{configurations}, endowed with the product topology. A basis of open sets for this topology is the set of cylinders: for $n \in \N^*$, $u\in \A^n$ and $\U = \{i_1, \dots, i_n\} \subset \Z$, we define $[u]_\U$ as 
\[
	[u]_\U = \{ x\in \X \mid \forall j \in \llbracket 1, n \rrbracket, x_{i_j} = u_j \}.
\] 
This topology makes $\X$ compact and metrizable. For a given configuration $x\in\X$, $i\in \Z$ and $r \in \N$, we denote by $x_{i + \llbracket -r, r \rrbracket}$ the $(2r+1)$-uple $(x_{i-r}, x_{i-r+1}, \dots, x_{i+r})$.

A function $F$ on $\X$ is a \define{Cellular Automaton} (CA) of \define{radius} $r\in\N$ if there exists a local rule $f : \A^{2r+1} \to \A$ such that
\[ 
\forall x\in\X, \forall i\in\Z, \quad F(x)_i = f(x_{i + \llbracket -r, r \rrbracket}).
\]
A probability kernel $\Phi$ is a \define{Probabilistic Cellular Automaton} (PCA) of radius $r$ if there exists a stochastic matrix (also called the local rule) $\varphi : \A^{2r+1} \times \A \to [0,1]$ such that 
\[
\forall n\in \N^*, \forall u \in \A^n, \forall \U = \{i_1, \dots, i_n\} \subset \Z, \quad \Phi(x, [u]_\U) = \prod_{j=1}^{n} \varphi( x_{i_j + \llbracket -r, r \rrbracket}, u_j).
\]

Given a CA $F$ and $\epsilon >0$, a PCA $\Phi$ is called an \define{$\epsilon$-perturbation of $F$} if they have the same radius and their respective local rule verify $\forall a \in \A^{2r+1}, \, \varphi(a, f(a)) \geq 1-\epsilon.$
Denote by $F_\epsilon$ (of local rule $f_\epsilon$) the perturbation of $F$ by a uniform noise of size $\epsilon$, defined by:
\[
	\forall a\in \A^{2r+1}, \forall b\in \A, \quad f_\epsilon (a,b) = \begin{cases}
		1 - \epsilon + \frac{\epsilon}{|\A|} & \text{if }b = f(a) \\
		\frac{\epsilon}{|\A|} & \text{otherwise}
	\end{cases}.
\]
It is obviously an $\epsilon$-perturbation of $F$.

Denote by $\M(\X)$ the set of probability measures on $\X$. By compactness of $\X$ this set is also compact and metrizable. The respective actions of a CA $F$ and a PCA $\Phi$, defined for $\mu\in \M(\X)$ and observable $B$ as 
\begin{align*}
F\mu(B) &= \mu(F^{-1} B) \\
\Phi\mu(B) &= \int \Phi(x,B) \diff \mu(x)
\end{align*}
are continuous on $\M(\X)$. By standard ergodic theory arguments, the sets of invariant measures $\M_F = \{ \mu\in\M(\X) \mid F\mu = \mu \}$ and $\M_\Phi$ are not empty.

A PCA is said to be \define{ergodic} if  it has a unique invariant measure $\pi$ that attracts every initial measure $\mu\in\M(\X)$, in the sense that the sequence $\left(\Phi^n\mu\right)_{n\in\N}$ converges weakly toward $\pi$ (i.e. $\Phi^n\mu([u])\underset{n\to\infty}{\longrightarrow}\pi([u])$ for any word $u$). Thus the set of invariant measure is a singleton and the unique measure is attracting.

A PCA $\Phi$ is said to have positive rates if for all $a\in\A^{2r+1}$ and $b\in \A$, $f(a,b)>0$: any neighborhood can give all results, with positive probability. For example, $F_\epsilon$ has positive rates for all $\epsilon > 0$.

\subsection{Interpretation of the G\'acs paper}

In this article we are going to use a direct  consequence of the main result of \cite{Gacs01}. Indeed, it is possible to say that the Gács cellular automaton is ergordic if the distribution of errors is quite near a uniform perturbation under a certain small rate.  If $F$ is a fixed CA, for  a trajectory under the action of the perturbed cellular automaton denoted $(x^t)_{t\in\N}$, we say that there is an \define{error} at cell $(i,t) \in \Z^2$, or at cell $i$ at time $t$, if $x^t_i \neq F(x^{t-1})_i$.

\begin{thm}
\label{thm:Gacs}
There exist n finite alphabet $\G$, a CA $G:\G^\Z\to\G^\Z$  of radius $1$ and a rate $\epsilon_c>0$ verifying the following property: if perturbed by a noise such that for all finite $S\subset \Z$ and all events $H$ in the past,
\[
P(\text{error in each cell of }S \mid H) \leq \epsilon_c^{|S|},
\]
then the associated stochastic dynamical system is not  ergodic.
\end{thm}
Observe that the condition on the noise is more general than our definition of an $\epsilon$-perturbation. In particular it is always verified for $G_\epsilon$, the perturbation of $G$ by a uniform noise of size $\epsilon < \epsilon_c$.

In the article \cite{Gacs01}, the author describes those admitted trajectories as created by an ``adversary" trying to defeat the conclusion of the theorem. This can sum up the idea behind the following CA: we couple $G$ with an adversary CA $F$ which can send errors on the trajectories of $G$. When perturbed by a uniform noise of size $\epsilon = \epsilon_2$, the adversary sends patches of errors (adding to the ones of the uniform noise) that are small enough that they can be corrected by $G$. For $\epsilon < \epsilon_1$ however, those additional errors are too large to be corrected by any CA with radius $1$.

\section{Description of the cellular automaton and main result}\label{section-Description}

Fix two parameters, $k\in\N$ and $a\in\N$. The alphabet is $\A=\G\times\F$, with $\G$ the alphabet of the G\'acs cellular automaton, and $\F$ is composed of all the integers from $0$ to $k-1$ and particles $\nearrow_s$, with $0\leq s< ak$. We denote by $0_\G$ a fixed symbol in $\G$ for the rest of the paper.

We first define a cellular automaton $F$ on $\F^\Z$ of radius $1$ as the following: for any configuration $z\in\F^{\Z}$ and $i\in\Z$,
\[
F(z)_i = \begin{cases}
         	\nearrow_{s+1} & \text{if } z_{i-1} = {\nearrow_s}, \text{ with } 0\leq s < ak \\
         	s + 1 \mod k & \text{if } z_{i-1} \notin \{{\nearrow_t} \mid 0\leq t < ak\} \text{ and } z_i = {\nearrow_s}  \\
         	z_i + 1 \mod k & \text{otherwise} \\
         \end{cases}
\]
with the convention ${\nearrow_{ak}} \coloneqq 0$. Essentially, the arrows are particles going to the right with speed $1$ which synchronizes the values of the cell they go through, with a limited lifetime of $ak$ time steps. On the other cells $F$ acts as the operation $+1$ on the finite group $\Z / k\Z$.

\begin{figure}[!h]
    \centering
    \includegraphics[scale=0.7]{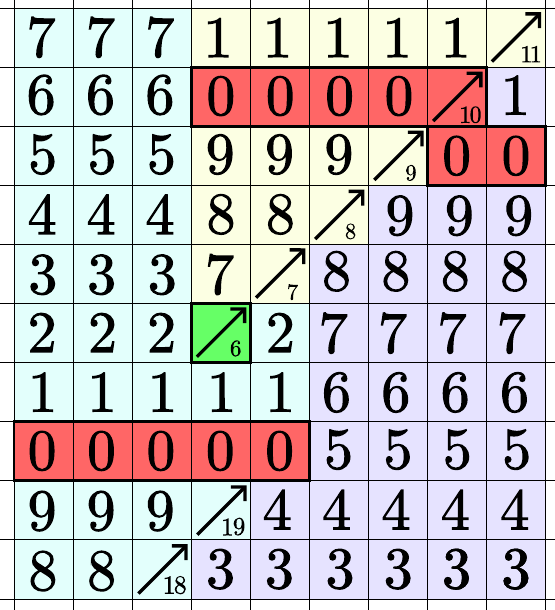}
    \caption{Some steps for $k=10$, $a=2$. Time goes up and only the $\mathcal{F}$ layer is shown. The pale colors represent the different synchronized zones created by the arrows. In green, an error creates a $\nearrow_6$. In red, the cells where a $0$ is projected onto the $\G$ layer.}
    \label{fig:GeneralSchema}
\end{figure}

Finally, let us define a 1-dimensional cellular automaton $T$ on $\A$ of radius $1$ as the following: for any configuration $x = (x_i)_{i\in\Z}$, with $x_i = (y_i,z_i) \in \A=\G\times\F$, we have
\[
T(x)_i = \begin{cases}
         	(G(y_{i-1}, y_i, y_{i+1}), F(z)_i) & \text{if }F(z)_i \notin \{0\}\cup\{\nearrow_s \mid s = 0 \mod k\} \\
         	(0_{\G}, F(z)_i) & \text{otherwise.}
         \end{cases}
\]
Intuitively, $T$ acts as the product cellular automaton of $G$ and $F$, except on the cells where $F$ gives a $0$ or a $\nearrow_s$ with $s=0\mod k$, at which point it ``projects" the $0$ onto the $\G$ layer.

\begin{thm}\label{thm:TransitionDePhase}
For all $\epsilon>0$, denote by $T_\epsilon$ the probabilistic cellular automaton defined as the perturbation of $T$ by a uniform noise on $\A$ of size $\epsilon$. There exists $k,a\in\N$ and $0 < \epsilon_1 < \epsilon_2 < \epsilon_3 < 1$ such that:
\begin{enumerate}
    \item for all $0 < \epsilon < \epsilon_1$, $T_\epsilon$ is uniformly ergodic (low noise regime);
    \item $T_{\epsilon_2}$ is not ergodic;
    \item for all $\epsilon >  \epsilon_3$, $T_\epsilon$ is uniformly ergodic (high noise regime).
\end{enumerate}
\end{thm}

\begin{rem}
The third point is a direct consequence of Proposition 3.6 of \cite{MST19}, and is true for any perturbed cellular automaton.
\end{rem}

Our intuition behind the main result is the following.
In the absence of noise, at each cell the $\F$-layer takes the value $0$ -- and thus the $\G$-layer is forced to take the value $0_\G$ -- once every $k$ time steps.
We think of this as simulating an error on the $\G$-layer.
The automaton $G$ can easily repair individual errors, but if a contiguous interval of more than $k$ cells is synchronized on the $\F$-layer, every $k$ steps their $\G$-layers are erased, preventing all flow of information across the interval.

\begin{figure}[h]
	\begin{centering}
	\begin{tabular}{|c|c|c|c|}
		\hline 
		 & $\epsilon=10^{-5}$ & $\epsilon=10^{-4}$ & $\epsilon=10^{-3}$\tabularnewline
		\hline 
		\hline 
		$a=50$ & \includegraphics[scale=0.35]{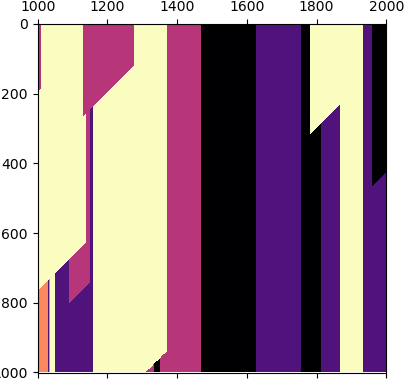} & \includegraphics[scale=0.35]{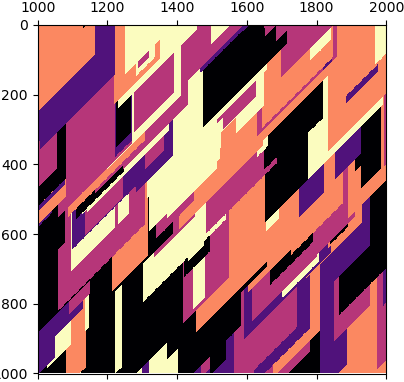} & \includegraphics[scale=0.35]{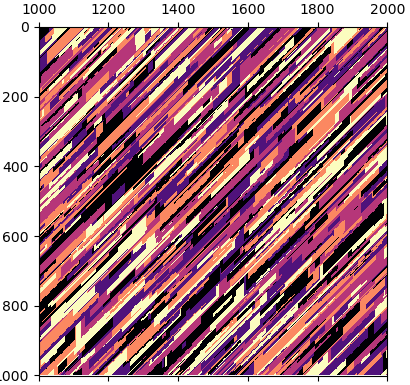}\tabularnewline
		\hline 
		$a=100$ & \includegraphics[scale=0.35]{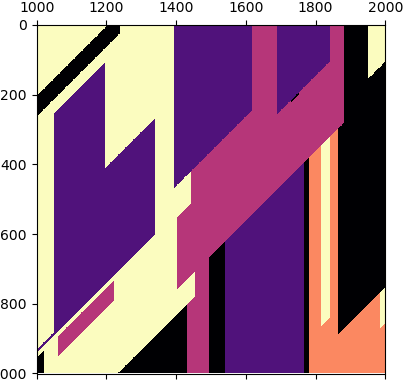} & \includegraphics[scale=0.35]{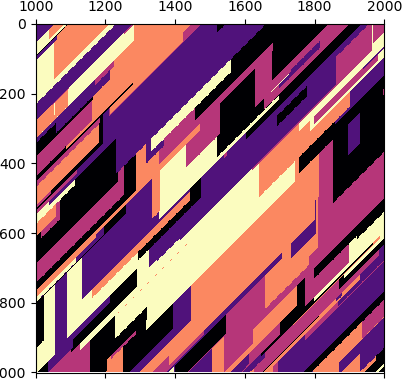} & \includegraphics[scale=0.35]{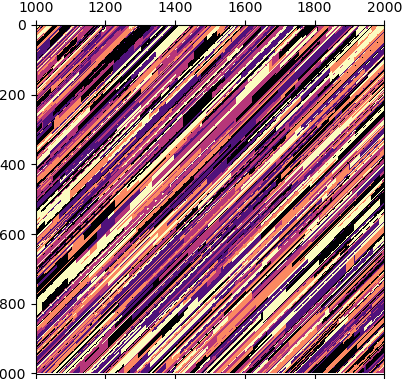}\tabularnewline
		\hline 
		$a=200$ & \includegraphics[scale=0.35]{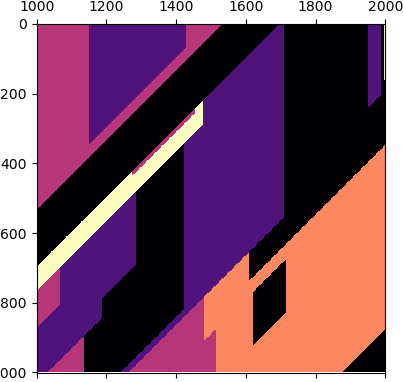} & \includegraphics[scale=0.35]{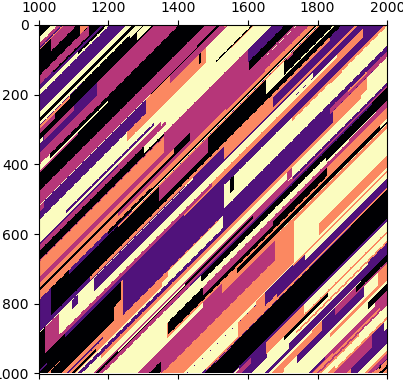} & \includegraphics[scale=0.35]{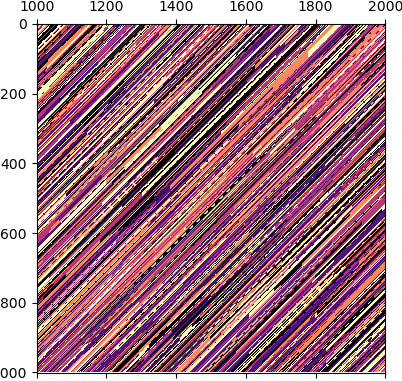}\tabularnewline
		\hline 
	\end{tabular}
	\par\end{centering}
	\caption{\label{fig:Illustrations-d'orbites}Illustrations of $\epsilon$-perturbed orbits of $F$, for $k=5$ and several values of $a$ and $\epsilon$. In each, we represent the subtraction of the value in a given cell by the time, mod $k$: the blocks of same color have synchronized values. A smaller $\epsilon$ increases the lifetime of the synchronized blocks, while a larger $a$ decreases it (there is more arrows) but increases their size (the lifetime of the arrows increases).}
\end{figure}

In the low noise regime $0 < \epsilon < \epsilon_1$, if the second parameter $a$ is large enough, random errors will create arrows with a relatively long lifetime, which produce long synchronized intervals that in turn persist long enough to prevent information flow and imply ergodicity.
At $\epsilon = \epsilon_2$, there is enough noise for the synchronized zones to de-synchronize and decorrelate in the $\F$ layer, producing small zones of errors.
Meanwhile, $\epsilon_2$ is small enough that for large $k$, the G\'acs CA $G$ has time to correct itself from such errors and retain information indefinitely, implying non-ergodicity.
Lastly, in the high noise regime $\epsilon > \epsilon_3$, there is so much noise that no information can navigate. We illustrate the effects of $a$ and $\epsilon$ in Figure \ref{fig:Illustrations-d'orbites}.

The parameters $k$ and $a$ and the rates $\epsilon_1$ and $\epsilon_2$ are chosen as follows.
We first let $k$ be large but arbitrary, and show in Section~\ref{sec:non-ergodicity} that there always exists an $\epsilon_2 > 0$, depending on $k$ but not $a$, such that $T_{\epsilon_2}$ is non-ergodic.
In Section~\ref{sec:low-noise-ergodicity}, we fix a large $k$, and show that there exists an $a$ such that $T_\epsilon$ is ergodic for all small enough $\epsilon$.

\section{Non-ergodicity}
\label{sec:non-ergodicity}

In this section, we look for an error rate $\epsilon_2$ such that $T_{\epsilon_2}$ is non-ergodic. To achieve this, we find a good parameter $k$ to use Theorem \ref{thm:Gacs} on the $\G$ layer: if $T_{\epsilon_2}$ is non-ergodic on this layer, then it is non-ergodic globally. We thus need to bound the probability of having an error on all sets $S$ of cells by $\epsilon_c^{|S|}$ for a fixed perturbation size $\epsilon_2$. There are two sources of errors on the $\G$ layer: the $0$ projected from the $\mathcal{F}$, and the errors from the noise itself. We can study these two sources relatively independently.

In this section the parameter $k$ is large and even but otherwise arbitrary, while $a$ is completely arbitrary.
We will find a suitable value for $\epsilon_2$ that depends on $k$ but not on $a$.

\subsection{Definitions and goals}

For a given $S\subset\Z$, define $Z_{S}$ as the event ``all the cells in $S$ have an $0$ in the $\mathcal{F}$-layer at time $t = 0$''. To show the non-ergodicity of $T_\epsilon$ for a suitable $\epsilon > 0$, we prove that for all finite $S\subset\Z$ and $H$ event in the past of the $\G$-layer, we have 
\[
P\left(Z_{S}\mid H\right)\leq \left(\frac{\epsilon_{c}}{2}\right)^{|S|}
\]
where $\epsilon_{c}$ is the critical value for the G\'acs CA from Theorem \ref{thm:Gacs}. We use $\frac{\epsilon_c}{2}$ for the bound to have some room to maneuver for the errors in the $\G$ layer at $t=0$, which we deal with in Section \ref{sec:errorsGLayer}.

We model the uniform $\epsilon$-perturbation $T_\epsilon$ as follows.
Let $(E^t_i)_{(i,t) \in \Z^2}$ be an ensemble of independent random variables $E^t_i$, each of which takes values in $\A \cup \{\bot\}$ with $P(E^t_i = \bot) = 1 - \epsilon$ and $P(E^t_i = a) = \epsilon / |\A|$ for each $a \in \A$.
We call it the \define{error field}, and the event $E^t_i \in \A$ is called a \define{potential error} at $(i,t)$.
Now the invariant measures of $T_\epsilon$ correspond exactly to the distributions on $(x^t_i, E^t_i)_{(i,t) \in \Z^2}$ that are invariant under vertical shifts and such that $x^t_i = T(x^{t-1})_i$ if $E^t_i = \bot$ and $x^t_i = E^t_i$ if $E^t_i \in \A$.

Let $H$ be an event in the past of the $\G$-layer, and let $(x^t)_{t \in \Z} = (y^t, z^t)_{t \in \Z}$ be a random trajectory of $T_\epsilon$ with error field $(E^t_i)_{(i,t) \in \Z^2}$. In order to bound $P(Z_{S}\mid H)$, we can suppose that more information is known.
Let $C = S \times [-k/2,-1]$ and $C' = S \times [-k/2,0]$, and let $K$ be an event in which the following data are fixed:
\begin{itemize}
	\item the contents on the $\F$-layer of all cells in $S \times \{-k/2\}$, which we call the \emph{initial configuration};
	\item the positions of all $0$-symbols on the $\mathcal{F}$-layer inside $C$;
	\item the positions (but not types) of all arrow symbols on the $\mathcal{F}$-layer inside $C$;
	\item the positions of all arrow symbols $\nearrow_s$ with $s = 0 \mod k$ inside $C$;
	\item the positions of all errors on the $\G$-layer (meaning cells $(i,t)$ such that $y^t_i \neq G(y^{t-1})_i$) inside $C$.
\end{itemize}
The number of such events is finite and they form a partition of the phase space, so if we can show $P(Z_S \mid H, K) \leq \left(\frac{\epsilon_{c}}{2}\right)^{|S|}$ for all such $K$, then
\[
P(Z_{S}\mid H) = \sum_K P(Z_{S}\mid H,K)P(K) \leq \left(\frac{\epsilon_{c}}{2}\right)^{|S|}\sum_K P(K) = \left(\frac{\epsilon_{c}}{2}\right)^{|S|}.
\]
For the sake of notational simplicity, we replace $H$ by $H \cap K$ for a fixed but arbitrary such $K$, and denote $P_H \coloneqq P({\cdot} \mid H)$.

Let $d \in \N$ be a parameter whose value we fix later.
Define an \emph{incident} as a cell of $C$ containing a $0$-symbol or an arrow on the $\mathcal{F}$-layer, or an error on the $\G$-layer.
The \emph{last incident} of a cell $s \in S$ is the incident $(s, t) \in C$ with the highest value of $t$, if one exists.
Denote by $C_s = \{s\} \times [t+1, 0]$ the \define{post-incident column} of $s$; if $s$ has no last incident, define $C_s = \{s\} \times [-k/2,0]$.
We decompose $S$ into $S=S_{1}\sqcup S_{2}\sqcup S_{3}\sqcup S_{4}\sqcup S_{5}$ as follows.
\begin{itemize}
    \item $S_{1}$ are the cells without a last incident.
    \item $S_{2}$ are the cells whose last incident is a $0$ on the $\mathcal{F}$-layer. 
    \item $S_{3}$ are the cells whose last incident is an arrow less than $d$ steps in the past.
    \item $S_{4}$ are the cells whose last incident is an arrow at least $d$ steps in the past.
    \item $S_{5}$ are the cells whose last incident is an error on the $\G$-layer.
\end{itemize}
If a cell satisfies more than one condition, we include it only in $S_i$ with the lowest possible $i$.

\begin{figure}[!h]
    \centering
    \includegraphics[scale=0.6]{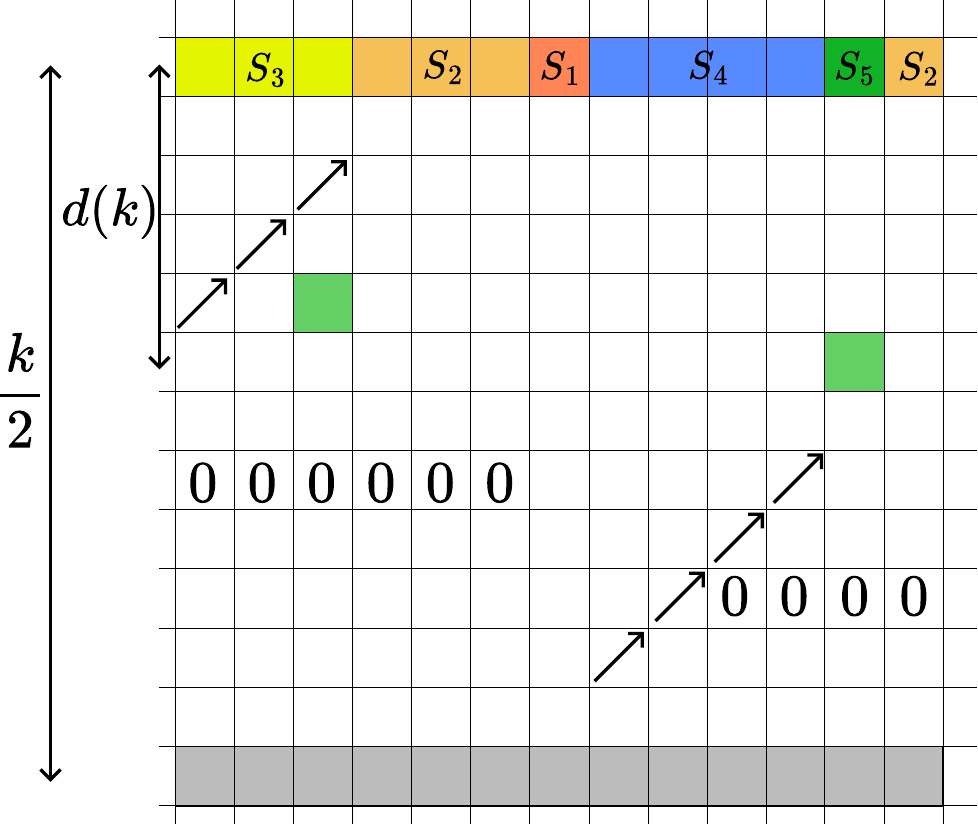}
    \caption{Decomposition of a given $S$ into each $S_i$. In grey, the known \define{initial configuration}. In green, the visible errors on the $\G$ layer.}
    \label{fig:SchemaS}
\end{figure}

Our goal is to find $\epsilon\coloneqq\epsilon(k)$ and $d \coloneqq d(k)$ such that, for all $i\in\left\{ 1,2,3,4,5\right\} $, there is an $\epsilon^{i} > 0$ such that $P_H\left(Z_{S_{i}}\right)\leq\left(\epsilon^{i}\right)^{\left|S_{i}\right|}$. Then, we would have 
\begin{align*}
P_H(Z_{S}) & =\sum_{\{ 1,2,3,4,5\}^{|S|}} \underset{\leq P_H(Z_{S_i})\text{ where }\left|S_{i}\right|\geq\frac{|S|}{5}}{\underbrace{P_H\left(Z_{S_{1}}\cap Z_{S_{2}}\cap Z_{S_{3}}\cap Z_{S_{4}}\cap Z_{S_{5}}\right)}}\\
	& \leq\max\left(\epsilon^{i}\right)^{\frac{|S|}{5}}\cdot5^{|S|}\\
	& \leq\left(5\max\left(\epsilon^{i}\right)^{\frac{1}{5}}\right)^{|S|}
\end{align*}
and we can prove the non-ergodicity of $T_\epsilon$ with Theorem \ref{thm:Gacs} if for all $i\in\left\{ 1,2,3,4,5\right\} $,  
\begin{equation}
5\left(\epsilon^{i}\right)^{\frac{1}{5}} \leq \frac{\epsilon_{c}}{2}.
\label{eq:epsilon_i}
\end{equation}

\subsection{Probability of a potential error}

The event $H$ can contain any information on what happened in the last $\frac{k}{2}$ iterations on the $\G$-layer. However, a transition which seems to respect the restriction on the $\G$-layer of the local deterministic rule of $T$ on a given cell does not necessarily mean that there were no errors. As we perturb the cellular automaton with a uniform noise, an error can place the expected symbol on the $\G$-layer and a random symbol on the $\F$-layer. As we study in this section the distribution of the $0$-symbols on the $\F$-layer, we need the probability of having an error on a cell given that its $\G$-value is correct.

The complete alphabet is $\A=\G\times\F$, with $\left|\F\right|=(a+1)k$. As the perturbation is uniform, when there is a potential error in a cell (with probability $\epsilon$), the symbol is chosen uniformly from $\A$. Therefore for a given cell, if we denote by $E$ (resp. ${VE}_{\G}$) the event ``there is a potential error on this cell'' (resp. ``an error on the $\G$ layer on this cell''), we have 
\begin{align*}
P\left(E\mid\overline{VE_{\G}}\right) & =\frac{P\left(E\cap\overline{{VE}_{\G}}\right)}{P\left(\overline{{VE}_{\G}}\right)}\\
	& \geq P\left(\overline{{VE}_{\G}}\mid E\right)P\left(E\right)\\
	& =\frac{\epsilon}{\left|\G\right|}.
\end{align*}

Thus, if there is no error in the $\G$-layer of a given cell, the probability of having no potential error in the $\mathcal{F}$-layer is less than $1-\frac{\epsilon}{\left|\G\right|}$. As the variables $E^t_i$ are mutually independent, the probability of not having any potential error on $n$ given cells is less than $\left(1-\frac{\epsilon}{\left|\G\right|}\right)^{n}$.

\subsection{Probability of a 0 on the $\F$ layer.}
For each case $S_i$, we study the probability of having $0$-symbols on the $\F$ layer of $S_i$ at time $t=0$.

\subsubsection{Cells of $S_{1}$: no last incident}

The $\F$-layer of each column $C_s = \{s\} \times [-k/2,0]$ for $s \in S_1$ is independent of the other columns of $C$, since no arrow crosses the column.
Hence we have $P_H\left(Z_{S_{1}}\right)=P_H(Z_s)^{\left|S_{1}\right|}$ for any $s \in S_1$.

The event $Z_s$ can occur for one of two reasons: either the column $C_s$ contains no potential errors and $s$ contains the symbol $k/2$ in the initial configuration, or $C_s$ contains at least one potential error and the last one, happening at time $t \in [-k/2,0]$, produces the symbol $t \mod k$ on the $\F$-layer.
Since $s \in S_1$, the last potential error $E^t_s$ of the column cannot produce an arrow and cannot cause a $0$ to occur on the $\F$-layer before time $0$.
There are no other constraints, so $E^t_s$ has at least $k/2$ possible values for the $\F$-layer that are consistent with $H$, each of which is equally likely, and hence the probability of producing $t \mod k$ is at most $2/k$.
Thus,
\begin{align*}
P_H(Z_s) & \leq P_H(\text{no potential errors in $C_s$})+P_H(\text{potential errors in $C_s$})\cdot\frac{2}{k}\\
	& \leq\left(1-\frac{\epsilon}{2\left|\G\right|}\right)^{\frac{k}{2}}+\frac{2}{k}.
\end{align*}
Thus $P_H(Z_s) \leq 2\max\left(\left(1-\frac{\epsilon}{2\left|\G\right|}\right)^{\frac{k}{2}},\frac{2}{k}\right)\eqqcolon\epsilon^1$ and $P_H\left(Z_{S_{1}}\right) \leq \left(\epsilon^1\right)^{|S_1|}$. Equation \ref{eq:epsilon_i} will be satisfied for $k$ large enough if $1/\epsilon = o(k)$, so that $\left(1-\frac{\epsilon}{2\left|\G\right|}\right)^{\frac{k}{2}}\tendsto k{\infty}0$. 

\subsubsection{Cells of $S_{2}$: $0$ on $\F$-layer}
As in the case of $S_1$, the post-incident column $C_s$ of each $s \in S_2$ is independent of the other columns of $C$.
For $s$ to contain a $0$ at time $t = 0$, we now need a potential error in $C_s$.
A similar but even simpler computation as in the $S_1$ case leads to
\[
P_H\left(Z_{S_{2}}\right)\leq\left(\frac{2}{k}\right)^{\left|S_{2}\right|}.
\]
and Equation \ref{eq:epsilon_i} will be satisfied for $k$ large enough.

\subsubsection{Cells of $S_{3}$: arrow less than $d$ steps ago}
A cell $s \in S_3$ has as its last incident an arrow at some time $t_s \in [-d(k)+1, 0]$.
We divide $S_3$ into (not necessarily contiguous) blocks of cells of size at most $d(k)$ as follows: two cells $s, s+m \in S_3$ are in the same block if $t_{s+m} = t_s + m$ and for each $j = 1, \ldots, m-1$ there is an arrow at $(s+j, t_s+j)$; thus, their last incidents are the same arrow at different times.
Note that the arrows' timers may not be consistent, as there might be an error between them that produces a new arrow in place of the old one.
For a block $B \subseteq S_3$, denote $C_B = \bigcup_{s \in B} C_s$.
The $\F$-layers of the sets $C_B$ for different blocks $B$ are independent of each other, as they come from different arrows.
Once the timer of the arrow that produces the block $B$ is fixed, the $\F$-layers of its own columns are mutually independent.

Take a block $B$.
We consider several possibilities and show that in each case, with high probability there exists at least one coordinate $(i,t) \in C$ such that when the other $E$-variables are kept constant, the conditional probability for $E^t_i$ to have a value that results in $Z_B$ is at most $8/k$.
In particular, $E^t_i = \bot$ will be inconsistent with either $H$ or $Z_B$.

Let $(j_i, u_i)$ for $i = 0, \ldots, N$ be the maximal set of coordinates of $C$ that have arrows, contain the last incident of each element of $B$, and satisfy $j_{i+1} = j_i+1$ and $u_{i+1} = u_i+1$ for all $0 \leq i < N$.
This is the entire known path of the arrow of $B$.

For some $(j_i, u_i)$ on the path, there exists $h \geq 1$ such that $(j_i, u_i+h)$ contains a 0 on the $\F$-layer, but none of $(j_i, u_i+h')$ for $1 \leq h' < h$ contains a 0, an arrow, or an error on the $\G$-layer.
For each such position, we condition on whether any of the coordinates $(j_i, u_i+h')$ for $1 \leq h' \leq h$ contain potential errors.
If any of them do, we can ignore the 0 at $(j_i, u_i+h)$, since it does not affect the timer of the arrow at $(j_i, u_i)$.
If none of them do, then we know the timer of the arrow at $(j_i, u_i)$ modulo $k$, and we know that it is not $u_i \mod k$.
Then the arrow at $(j_i, u_i)$ becomes a \define{good known arrow}.
Any $(j_i, u_i)$ that contains an arrow with a timer that is $0 \mod k$ is also a good known arrow, as is an arrow at time $-k/2$ unless its timer is $-k/2 \mod k$.

The good arrows partition the path into independent sub-paths that we analyze separately.
Each sub-path begins at a good known arrow or the start of the path, and ends just before the next one, or at the end of the path.

\begin{itemize}
\item[Case 1]
  Suppose there exists a sub-path that either begins at a good known arrow or is the first sub-path that begins after time $-k/2$, and contains the last incident $(s, t_s)$ of some $s \in B$.
  For $Z_s$ to hold, there must be a potential error either on the column $C_s$, or between the start of the sub-path and $(s, t_s)$.
  Note that it might be the potential error at $(j_0, u_0)$ that creates the arrow.

  \begin{itemize}
  \item[Case 1.1]
    Suppose the column $C_s$ contains a potential error.
    The value of the last potential error of $C_s$ determines whether $Z_s$ holds, and does not affect any other part of the configuration.
    It has at least $ak/2-k/2$ values that are consistent with the current context.
    With the same computation as before, the probability of $Z_s$ is at most $4/k$.
    
  \item[Case 1.2]
    Suppose $C_s$ contains no potential errors.
    For $Z_s$ to hold, there must be at least one potential error on the sub-path, on or before $(s, t_s)$.
    Let $E_{j_i}^{u_i}$ be the latest such potential error.
    Consider now the end of the sub-path.

    \begin{itemize}
    \item[Case 1.2.1]
      If the sub-path extends to time $t = 0$, then as before, $E_{j_i}^{u_i}$ has at least $ak/2 - k/2$ values that are consistent with the current context, and the probability of $Z_s$ is at most $4/k$.
    \item[Case 1.2.2]
      Otherwise the sub-path ends in a known good arrow, or is the last sub-path and does not extend to time $t = 0$.
      Suppose now that there is a potential error at the endpoint of the path (which destroys the arrow or produces the known good arrow).
      Then, as before, $E_{j_i}^{u_i}$ has at least $ak/2 - k/2$ consistent values and we are done.
    \item[Case 1.2.2]
      Suppose then that the endpoint does not contain a potential error.
      Then the value modulo $k$ of the timer of the arrow produced by $E_{j_i}^{u_i}$ is fixed, since there are not potential errors between $E_{j_i}^{u_i}$ and the end of the sub-path, where either the value of the timer modulo $k$ is known or the arrow dies naturally.
      In both cases, the timer is not $u_i \mod k$, so $Z_s$ surely does not hold.
    \end{itemize}
  \end{itemize}

  \item[Case 2]
    If Case 1 does not hold, then the path begins at time $t = -k/2$ with an arrow whose timer is $-k/2 \mod k$, and there are no good known arrows between this point and the last incident of any element of $B$.
    Let $E_{\mathrm{path}}$ be the event that there is a potential error on the path of the arrow between time $t = -k/2$ and $t = -d(k)$.
    If $E_{\mathrm{path}}$ holds, then considering the last such potential error, we once again obtain an upper bound of $4/k$.
    The converse case is handled by the following claim.
    
    \begin{claim} \label{clm:conditionalProbability}
      If $d(k) \leq k/4$ and $\epsilon(k) \geq 2 |\G| \left(1 - (4/k)^{4/k} \right)$, then $P_H(\overline{E_{\mathrm{path}}} \mid \mathrm{Case~2}) \leq 4/k$.
    \end{claim}
    
    Supposing that the preconditions hold, the claim directly implies that the conditional probability of $Z_B$ in Case 2 as at most $8/k$.
\end{itemize}

\begin{proof}[Proof of Claim \ref{clm:conditionalProbability}]
  Since Case 2 implies $u_0 = -k/2$, the event $\overline{E_{\mathrm{path}}}$ is the intersection of $E_{j_i}^{u_i}=\bot$ for all $1 \leq i\leq k/2-d\left(k\right)$.  We then have
\[
  P_{H}\left(E_{\mathrm{path}} \mid \mathrm{Case~2}\right) =
  \prod_{i=1}^{k/2-d(k)} P_{H}\left(E_{j_i}^{u_i}=\bot \mid \bigcap_{1 \leq l < i} E_{j_l}^{u_l}=\bot, \mathrm{Case~2}\right).
\]

Let $i\in\N$ and suppose there are no potential errors on $(j_l, u_l)$ for $1 \leq l < i$.
Recall that in Case 2, every 0-symbol above the path has a potential error between it and the path, and hence does not constrain the timers on the path.
For a potential error on $(j_i, u_i)$ to be consistent with the known context even without another potential error on the path, it suffices to create a value $\nearrow_{m}$ with $(m \mod k)\leq\frac{k}{2}$: in the worst case, there is no $0$-symbol on the $\F$-layer of any column above the path. Thus,
\[
  P\left(E_{j_i}^{u_i}\in\A\mid\bigcap_{l<i}E_{j_l}^{u_l}=\bot, \mathrm{Case~2}\right) \geq \frac{\epsilon}{\left|\G\right|}\cdot\frac{a\frac{k}{2}}{ak}=\frac{\epsilon}{2\left|\G\right|}.
\]

Applying the above result to all $i$, we get 
\[ P_{H}\left(E_{\mathrm{path}} \mid \mathrm{Case~2}\right)
  \leq
  \left(1-\frac{\epsilon}{2\left|\G\right|}\right)^{\frac{k}{2}-d\left(k\right)} \leq \left(1-\frac{\epsilon}{2\left|\G\right|}\right)^{\frac{k}{4}}
  \leq 4/k.
\]
Here, the last two steps follow from the assumptions on $d(k)$ and $\epsilon(k)$.
\end{proof}

Finally, we have for each block $B$ the bound $P_H(Z_B) \leq \frac{8}{k}$. The values on the $\F$ layer being mutually independent with the other blocks, and all of them having size at most $d(k)$, we obtain
\[
	P_H(Z_{S_3}) \leq P(Z_B)^{\frac{|S_3|}{d(k)}} \leq \left(\left(\frac{8}{k}\right)^{\frac{1}{d(k)}} \right)^{|S_3|}
\]
and Equation \ref{eq:epsilon_i} will be satisfied for $k$ large enough if $d(k)=o(\ln k)$.

\subsubsection{Cells of $S_{4}$: arrow at least $d$ steps ago}

The idea is to take $d(k)$ large enough so that we can make the same computations as in $S_{1}$. For each $s\in S_4$, the potential errors in the columns post-incident are mutually independent. As there is no $0$ symbols after the arrow's passage, there can be a $0$-symbol in $s\in S_4$ only if there was no potential error after the arrow, or if the last potential error gave the right symbol. As is must be the case on all $s\in S_4$, the same analysis as before gives 
\[
	P_H\left(Z_{S_{4}}\right)\leq\left(2\max\left(\left(1-\frac{\epsilon}{\left|\G\right|}\right)^{d(k)},\frac{2}{k}\right)\right)^{\left|S_{4}\right|}.
\]
Equation \ref{eq:epsilon_i} will be satisfied for $k$ large enough if $1/d(k)=o(\epsilon)$ so that $\left(1-\epsilon\right)^{d(k)}\tendsto k{\infty}0$. 

\subsubsection{Cells of $S_{5}$: error on $\G$-layer}
Once again, each post-incident column for $s\in S_5$ are mutually independent of each other as no arrow goes through them. As the last incident was an error in the $\G$ layer, there was also an error at this step in the $\F$ layer (which did not produce an arrow nor a $0$). We then have the same computations as for $S_2$, the probability that the last error produced the right symbol is bonded by $\frac{2}{k}$. Thus, 
\[
	P_H\left(Z_{S_{5}}\right)\leq \left(\frac{2}{k}\right)^{\left|S_{5}\right|}.
\]
and Equation \ref{eq:epsilon_i} will be satisfied for $k$ large enough.

\subsubsection{Conclusion}

In order to satisfy Equation \ref{eq:epsilon_i} for all $i \in \{1,2,3,4,5\}$, we have the following sufficient conditions on $\epsilon$ and $d$, when $k$ goes to $+\infty$: 
\begin{itemize}
    \item $1/\epsilon(k)=o(k)$ and $1/d(k)=o(\epsilon(k))$: we need enough noise to be sure to have errors on each colmuns with high probability.
    \item $\lim_{k\to\infty}\epsilon(k)=0$ and $\lim_{k\to\infty}d(k)=+\infty$.
    \item $d(k)=o(\ln k)$ and $\frac{d(k)}{\epsilon(k)}=o(k)$: we need the synchronized blocks to not be too large.
    \item $d(k) \leq k/4$ and $\epsilon(k) \geq 2 |\G| \left(1 - (4/k)^{4/k} \right)$: we need these to apply Claim \ref{clm:conditionalProbability}.
\end{itemize}
A solution would be $d(k)\coloneqq\ln\ln k$ and $\epsilon(k)\coloneqq\frac{1}{\ln\ln\ln k}$, and we have the result for the noise $\epsilon(k)$, for a $k$ large enough.
For the last condition, this follows from $1 - (4/k)^{4/k} \sim 4/k \ln(k/4)$ for large $k$.

\subsection{The errors on the $\G$ layer}
\label{sec:errorsGLayer}
We add here the second source of errors on $S$ at time $0$, the noise itself. By independence of the perturbation, if $E_S$ denotes the event ``having errors due to noise on all cells of $S$", we have
\begin{align*}
P_H(\text{all errors in }S) & = \sum_{R\subset S} P_H(Z_{R}\cap E_{S\backslash R}) \\
	& =\sum_{R\subset S} P_H(Z_{R})P_H(E_{S\backslash R}) \\
	& \leq\sum_{R\subset S} \left(\frac{\epsilon_{c}}{2}\right)^{|R|} \cdot \epsilon^{|S|-|R|}
\end{align*}
so as long as $\epsilon(k) \leq \frac{\epsilon_{c}}{2}$, we have $P_H(\text{all errors in S})\leq \left(\frac{\epsilon_{c}}{2}\right)^{|S|}\cdot2^{|S|}=\epsilon_{c}^{|S|}$.

Finally, with $\epsilon(k) = \frac{1}{\ln\ln\ln k}$ and $d(k) = \ln\ln k$, we can fix $k$ large enough such that $\epsilon(k)\leq \frac{\epsilon_c}{2}$ and Equation \ref{eq:epsilon_i} is verified for each $i\in \{1,2,3,4,5\}$. We can then use Theorem \ref{thm:Gacs} to conclude that at $\epsilon_2 = \epsilon(k)$, the PCA $T_\epsilon$ is not ergodic.

\section{Ergodicity in the low-noise regime}
\label{sec:low-noise-ergodicity}

In this section, we want to show that $T_\epsilon$ is ergodic when $\epsilon$ is small enough. For the $\G$ layer, we only use the fact that it is of radius $1$, and show that no information can go through because of the projection of $0_\G$ symbols. As it is illustrated in Figure \ref{fig:Illustrations-d'orbites}, when $a$ increases and $\epsilon$ decreases, the synchronized zones are larger and stay synchronized for a longer time.

For that, consider a trajectory simulated arbitrarily far in the past. We introduce two random walks $L_{n}$ (Section \ref{subsec:L_n}) and $R_{n}$ (Section \ref{subsec:Bordure-Droite}) that bound in a certain sense $l_t$ and $r_t$, the borders of the dependence cone of the cell at position $0$ at time $0$, $t$ steps in the past (the cells whose values can influence the value of the cell at $(0,0)$). We then show that almost surely, $L_{n}\tendsto n{\infty}+\infty$ and $R_{n}\tendsto n{\infty}-\infty$, and thus $l_t$ and $r_t$ ``cross" (Proposition \ref{prop:Croisement p.s.}): the value at $0$ at time $0$ is independent of the starting configuration and depends only on the noise, the simulated probabilistic cellular automata is ergodic.

\subsection{Update maps and dependence cones}

To show the ergodicity of $T_{\epsilon}$, we use the notion of update maps as in~\cite{MST19}. We define it here for a PCA of radius $1$. All random variables in the rest of the article are defined on a common universe $\Omega$.
\begin{defn}
$\psi: \A^3\times[0,1]\to\A$ is a local update map of $T_\epsilon$ if 
\[
	P\left(\psi(a_{-1}a_0a_1, U) = b\right) = f_\epsilon (a_{-1}a_0a_1,b)
\]
where $U$ is uniformly distributed on $[0,1]$ and $f_\epsilon$ is the local rule of $T_\epsilon$. The global update map associated is the function $\Psi: \A^\Z \times [0,1]^\Z \to \A^\Z$ defined by 
\[
	\Psi(x;u)_k = \psi(x_{k-1}x_kx_{k+1}, u_k).
\]
Its iterations can be recursively defined by 
\[
	\Psi^{t+1}(x; u^1, \dots, u^{t+1}) = \Psi(\Psi^t(x; u^1,\dots, u^t) ; u^{t+1}),
\]
so that $\Psi^t(x; U^1, \dots, U^t)$ (with $U^t_i$ a family of uniformly distributed on $[0,1]$ independent random variables) is distributed according to $T_\epsilon^t(x, \cdot)$.
\end{defn}

In the following we suppose that such a map $\Psi$ is fixed, along with a family $(U^{-t}_{i})_{t\in\N,i\in\Z}$. From \cite{MST19} we take the following criterion for the ergodicity of $T_\epsilon$ using coupling from the past:
\begin{prop}[{\cite[Proposition 3.3]{MST19}}]
Define $p_t(T_\epsilon) = P\left(x \mapsto \Psi^t\left(x; U^{-t}, \dots, U^0  \right)_0 \text{ is constant}\right)$.

If $p_t(T_\epsilon) \tendsto{t}{+\infty}1$, then $T_\epsilon$ is uniformly ergodic.
\end{prop}

\subsubsection{Dependence cone}

To study $p_t(T_\epsilon)$, we use the dependence cone for semi-configurations. For $m\in\Z$, as $T_\epsilon$ is of radius $1$, for a fixed family $(U^{-t}_i)_{i,t}$ the value $\Psi^t(x; U^{-t}, \dots, U^0)_{[m;+\infty)}$ only depends on the value of $x_{[m-t,+\infty)}$. To simplify the notations, define $\Psi^t_{m^+} : \A^\N \times [0,1]^{\Z^2} \to \A^\N$ as
\[
	\Psi^t_{m^+} (\alpha; u) = \Psi^t \left(x; u^{-t}, \dots, u^0 \right)_{[m,+\infty)}
\]
for any  $x\in\A^\Z$ such that $x_{[m-t,+\infty)} = \alpha$. Similarly, define $\Psi^t_{m^-} : \A^{\Z^-} \times [0,1]^{\Z^2} \to \A^{\Z^-}$ as
\[
	\Psi^t_{m^-} (\alpha; u) = \Psi^t \left(y; u^{-t}, \dots, u^0 \right)_{(-\infty,m]}
\]
for any  $y\in\A^\Z$ such that $y_{(-\infty,m+t]} = \alpha$.

\begin{defn}
For a given $\omega\in\Omega$, the left border of the dependence cone at time $-t$ is 
\[
	l_t(\omega) \coloneqq \min\left\{n\geq -t \mid \exists \beta \in \A^\N, \, \begin{array}{ccl}
	\A^{t+n+1} & \to & \A^\N \\
	\alpha & \mapsto & \Psi^t_{0^+}(\alpha\beta,U(\omega))
\end{array} \text{ is not constant}\right\}
\]
while the right border is 
\[
	r_t(\omega) \coloneqq \max\left\{ n \leq t \mid \exists \alpha \in \A^{\Z^-}, \, \begin{array}{ccl}
	\A^{t-n+1} & \to & \A^{\Z^-} \\
	\beta & \mapsto & \Psi^t_{0^-}(\alpha\beta,U(\omega))
\end{array} \text{ is not constant}\right\}
\]
with the conventions $\min(\emptyset) = +\infty$ and $\max(\emptyset) = -\infty$.
\end{defn}

We can interpret $[l_t,+\infty)$ as the cells that can influence the result of the simulation after $t$ steps on the right semi-configuration, as illustrated on Figure \ref{fig:l_t}. 

\begin{figure}[!h]
\centering{}\includegraphics[scale=0.6]{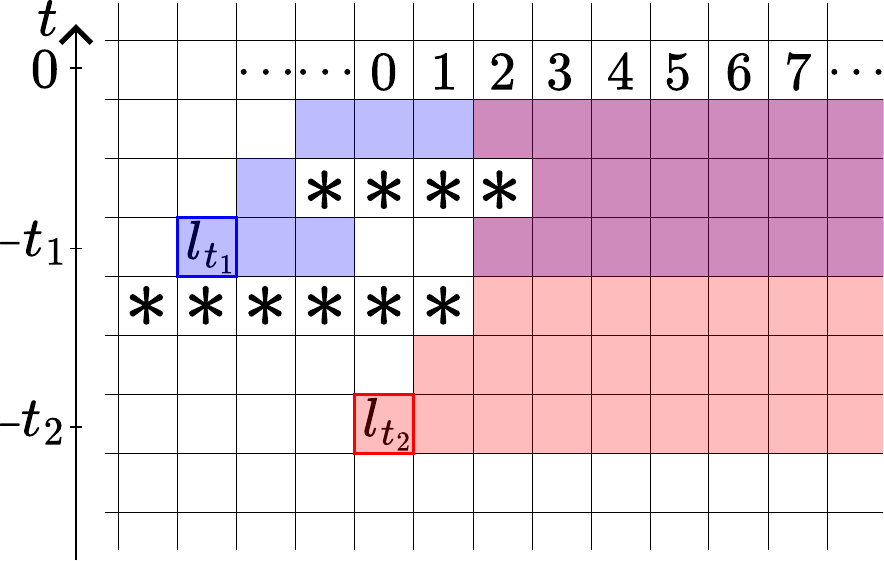} \caption{Illustration of $l_{t}$. The stars represents the errors, where the values on these cells depend only on the noise. In blue, the information flow from the cells in the configuration at $t=-t_{1}$ that can influence the values on the right semi-configuration at $t=0$. In red, the information flow for the configuration at $t=-t_{2}$. The cells at position $\left(l_{t_{1}},-t_{1}\right)$ and $\left(l_{t_{2}},-t_{2}\right)$ are framed. \label{fig:l_t}}
\end{figure}

\begin{lem}
For $\omega\in\Omega$ and $t\in\N$ fixed, if $l_t > r_t$ then $x\mapsto \Psi^t \left(x ; U(\omega)\right)_0$ is constant.
\label{lem:emptyDependenceCone}
\end{lem}
\begin{proof}
Fix two configurations $x = x_{(-\infty; l_t]} x_{[l_t+1;+\infty)}$ and $y = y_{(-\infty; l_t]} y_{[l_t+1;+\infty)}$. We then have
\begin{align*}
\Psi^t \left(x ; U(\omega)\right)_0 &= \Psi^t \left(x_{(-\infty; l_t]} x_{[l_t+1;+\infty)} ; U(\omega)\right)_0\\
	&= \Psi^t \left(y_{(-\infty;  l_t]} x_{[l_t+1;+\infty)} ; U(\omega)\right)_0\\
	&= \Psi^t \left(y ; U(\omega)\right)_0
\end{align*}
where the first swap is by definition of $l_t$ (and $0\in [0,+\infty)$), and the second by definition of $r_t$  (and $0\in (-\infty,0])$) and $r_t < l_t$.
\end{proof}

As the cellular automaton we study is of radius $1$, the information cannot go faster than $1$ cell per step: this is what the following lemma illustrates.

\begin{lem}
For all $\omega\in\Omega$ and $t,s\in\N$, $l_{t+s} \geq l_t -s$ and $r_{t+s}\leq r_t +s$.
\label{lem:passageOfTime}
\end{lem}
\begin{proof}
Without loss of generality, we can restrict ourselves to the case $s=1$. As the result is immediate for $l_t = -t$, we can suppose that $l_t > -t$ and fix a $n$ such that $-t-1 \leq n < l_t -1$ and $\beta\in\A^\N$. 

For $\alpha, \alpha^\prime \in \A^{t+n+2}$, fix $x,y\in\A^\Z$ such that $x_{[-t-1,+\infty)} = \alpha\beta$ and $y_{[-t-1,+\infty)} = \alpha^\prime\beta$. As $x$ and $y$ agree on $[n; +\infty)$,
\[
	\Psi(x; U^{-t-1}(\omega))_{[n+1;+\infty)} = \Psi(y; U^{-t-1}(\omega))_{[n+1;+\infty)} \eqqcolon \delta(\omega)
\]
so we can denote $\gamma,\gamma^\prime \in \A^{t+n+1}$ such that $\Psi(x; U^{-t-1}(\omega))_{[-t,+\infty)} = \gamma\delta$ and $\Psi(y; U^{-t-1}(\omega))_{[-t,+\infty)} = \gamma^\prime\delta$.

Thus by definition of $l_t$,
\begin{align*}
\Psi^{t+1}_{0^+}(\alpha\beta; U(\omega)) &= \Psi^{t}_{0^+}(\gamma\delta; U(\omega))\\
	&= \Psi^{t}_{0^+}(\gamma^\prime\delta; U(\omega))\\
\Psi^{t+1}_{0^+}(\alpha\beta; U(\omega)) &= \Psi^{t+1}_{0^+}(\alpha^\prime\beta; U(\omega)).
\end{align*}

Therefore, for all $n \in [-t-1; l_t -1 )$ and $\beta \in \A^\N$, the function $\begin{array}{ccl}
	\A^{t+n+1} & \to & \A^\N \\
	\alpha & \mapsto & \Psi^t_{0^+}(\alpha\beta,U)
\end{array}$ is constant. The proof is analog for $r_t$.
\end{proof}

\subsection{Markov Additive Chains}

To bound $l_{t}$ and $r_{t}$ (and force them to verify the hypotheses of Lemma \ref{lem:emptyDependenceCone}), we define two auxiliary random walks $L_{n}$, $R_{n}$.
The sequence $L_{n}$ satisfies an approximate version of $L_{n} \leq l_{t_{n}}$ for some sequence $t_{n}$ such that $t_{n+1} - t{n}$ is bounded (Proposition \ref{prop:Ln-lt}), and $L_{n} \tendsto{n}{\infty} +\infty$ a.s.
The sequence $R_{n}$ satisfies the symmetric criterion.

\subsubsection{High-level explanation}

The behavior of the chains $L_{n}$ and $R_{n}$ is governed by the positions of large synchronized zones on the $\F$-layer.
The pair $(L_{n}, t_{n})$ represents a position in a random trajectory $(x^t)_{t \in \Z^-}$ of the CA $T_\epsilon$.
To determine how it evolves, we check whether the $\F$-layer of the segment $[L_{n}-2k, L_{n}]$ is synchronized by an arrow created at time $t_{n}-2k-m$ for some (potentially large) $m \geq 0$, and whether this synchronization survives until time $t_{n}$.
If this is the case, the synchronization blocks all information flow from the left to $(L_{n}, t_{n})$ on the $\G$-layer, and we can set $L_{n+1} = L_n + s$ for some $s > 0$ that depends on the extent of the synchronized segment (that is, how far to the right of $L_{n}$ it reaches).
On the other hand, if such a synchronized segment does not exist, or if its synchronization is broken by errors before time $t_{n}$, we set $L_{n+1} = L_n - 4ak$.

In both cases the jump must be large enough to ensure that the next check is independent of the last one.
For a move to the right, this is rather simple.
A move to the left, on the other hand, typically happens because for some $m \geq 0$ (which we allow to be large in order to increase the probability of a move to the right) the segment $x^{t_{n}-2k-m}_{[L_{n}-2k, L_{n}]}$ is synchronized by an arrow, but this synchronization is broken by errors between time $t_{n}-2k-H$ and $t_{n}$.
As we have now conditioned on the contents of the region $B = [L_{n}-2k, L_{n}] \times [t_{n}-2k-m, t_{n}]$, in order to maintain independence we cannot allow the Markov chain $(L_{n}, t_{n})$ to pass through it.
Now, instead of setting $L_{n+1} = L_n - r$ for a large enough $r > 0$ that the chain cannot possibly return to $L_{n}-2k$ before time $-n-2k-m$, we pick $r = 4ak$ and allow the chain to remain close to the left border of $B$, replacing moves to the right by stationary steps until time $t_{n}-2k-m$.
For this, the chain $L_{n}$ must ``remember'' whether it has recently moved to the left.
See Figure \ref{fig:MAC} for a diagram of the evolution of $L_{n}$.

The chain $R_{n}$ behaves more or less symmetrically.

\subsubsection{Markov Additive Chains}

A bivariate Markov chain $(L_n,J_n)$ on $\Z\times E$ with $E$ discrete is called a Markov Additive Chain (MAC) if the \define{phase}  $(J_n)$ is a Markov chain on $E$ with transition matrix $\mathbf{P}$ influences the increments of the \define{level} $(L_n)$ in the following sense: for any $e\in E$ and $T\in\N$, the law of $(L_{T+n} - L_n, J_{T+n})$ given $\{J_T = e\}$ is independent of $((L_0, J_0), \dots, (L_T, J_T))$ and has the same distribution as $(L_n - L_0, J_n)$ given $\{J_0 = e\}$. In other terms: the future steps of the level only depends on the present phase.

In the remainder of the article, we suppose $(J_n)$ to be a Markov chain on $E = \{1,2\}$ with transition matrix $\mathbf{P} \coloneqq \begin{pmatrix}
\alpha_\epsilon & 1-\alpha_\epsilon \\ 
\beta_\epsilon & 1-\beta_\epsilon
\end{pmatrix} $.
Phase $2$ means that the chain $L_{n}$ has recently moved to the left.
This chain is ergodic with its unique invariant probability vector being $\pi = (\pi_1,\pi_2) \coloneqq \frac{1}{1-\alpha_\epsilon+\beta_\epsilon}(\beta_\epsilon,1-\alpha_\epsilon)$.

There is a version of the law of large numbers (see \cite{Asmussen03} chapter XI) for the level of a MAC: for any initial state, $\frac{L_n}{n}\tendsto{n}{\infty}\kappa^\prime$ a.s. where $\kappa^\prime$ is the \define{mean drift} defined as
\[
	\kappa^\prime = \sum_{i=1}^2 \sum_{j=1}^2 \pi_i \mathbf{P}_{i,j} \sum_{m\in\Z} m P\left(L_1 = m \mid J_0 = i, J_1=j\right).
\]
In particular if $\kappa^\prime > 0$ then $L_n \tendsto{n}{\infty} +\infty$ a.s..

The MACs defined in the following sections have the property that the distribution of the step $L_1 - L_0 = L_1$ given $J_1$ is independent of the initial phase $J_0$. In this case, $\kappa^\prime$ can be simplified as
\[
	\kappa^\prime = \pi_1 \sum_{m\in\Z} m P\left(L_1 = m \mid J_1=1\right) + \pi_2 \sum_{m\in\Z} m P\left(L_1 = m \mid J_1=2\right).
\]

\subsubsection{Definition of $L_n$}
\label{subsec:L_n}

Fix a numbering of the alphabet $\A = \{\alpha_0, \dots, \alpha_{|A|-1}\}$, and define the local update map $\psi$ of $T_\epsilon$ to be the following:

\[
	\psi\left( a_{-1}a_0a_1 ; u\right) = \begin{cases}
		T(a_{-1}a_0a_1) & \text{if } u < 1-\epsilon\\
		\alpha_i & \text{if } u \in \left[ 1-\epsilon+\frac{i}{|\A|}\epsilon, 1-\epsilon+\frac{i+1}{|\A|}\epsilon \right)
	\end{cases}.
\]

For $i\in\Z$ and $t\in\Z^-$, define the random variable $E_i^t$ with value in $\A\cup\{\bot\}$ as 
\[
	E_i^t = \begin{cases}
		\bot & \text{if } U_i^t < 1-\epsilon\\
		(\alpha_i)_{\F} & \text{if } U_i^t \in \left[ 1-\epsilon+\frac{i}{|\A|}\epsilon, 1-\epsilon+\frac{i+1}{|\A|}\epsilon \right)
	\end{cases}
\]
such that $E_i^t = \alpha\in\F$ if there is an error that gives the symbol $\alpha$ on the $\F$ layer at $(i,t)$, and $E_i^t=\bot$ when there is no error.

For $0 \leq s < ak$, we can define the event $A_i^t(s)$ (having an arrow $\nearrow_s$ at $(i,t)$) by 
\[
	A_i^t(s) = \bigsqcup_{m=0}^s \left[ \left(E_{i-m}^{t-m} = \nearrow_{s-m} \right) \cap \bigcap_{n=0}^{m-1} E_{i-n}^{t-n} = \bot \right]
\]
which has probability 
\[
	P\left(A_i^t(s) \right) = \sum_{m=0}^s \frac{\epsilon}{(a+1)k} (1-\epsilon)^m \underset{\epsilon\to0}{\sim} \frac{s+1}{(a+1)k}\epsilon.
\]

This allows us to define three disjoint events for a cell $(i,t)\in\Z^2$: $G_i^t$, $B_i^t$ and $O_i^t$. Which one occurs can be decided following the decision process described in Figure \ref{fig:GoodBadOther}.
These events will be used to determine the next move of $L_{n}$.
We will look at the segment $[L_{n}-2k, L_{n}]$ at time $t_{n}$ and check whether it is synchronized due to an arrow that passes through the spacetime position $(L_{n}-2k, t_{n}-2k)$.
If so, this is a ``good'' event and $L_{n}$ can make a move to the right.
If there is an arrow at that position, but the segment is not synchronized (due to an error or the arrow disappearing by itself), this is a ``bad'' event and $L_{n}$ will move to the left.
If there is no arrow at $(L_{n}-2k, t_{n}-2k)$, we repeat the check at $(L_{n}-2k, t_{n}-2k-1)$, then at $(L_{n}-2k, t_{n}-2k-2)$, and so on up to $(L_{n}-2k, t_{n}-2k-H)$ for a suitably chosen $H > 0$.
If none of these coordinates contain arrows (which is rare when $H$ is large), the chain moves to the left.
\begin{figure}[!h]
    \centering
    \begin{tikzpicture}[node distance=2cm]
		\node (start) [startstop] {At $(i,t)\in\Z^2$};
		\node (errorPath) [decision, fill=green!30, right of=start, xshift=1cm] {Error in the $2k-1$ green cells};
		\node (Ait) [decision, below of=errorPath] {$A_i^t(s)$};
		\node (s) [decision, below of=Ait, yshift=0.5cm] {$s \leq (a-2)k$};
		\node (triangle) [decision, fill=yellow!30, below of=s] {Error in the yellow triangle};
		\node (other) [startstop, right of=Ait, xshift = 0.5cm] {$O_i^t$};
		\node (bad) [startstop, above left of=s, xshift = -2cm, yshift=-0.25cm] {$B_i^t$};
		\node (good) [startstop, right of=triangle, xshift = 0.5cm] {$G_i^t$};
		
		\draw [arrow] (start) -- (errorPath);
		
		\draw [arrow] (errorPath) -- node[anchor=west] {no} (Ait);
		\draw [arrow] (errorPath) -- node[anchor=south east] {yes} (bad);
		
		\draw [arrow] (Ait) -- node[anchor=south] {no} (other);
		\draw [arrow] (Ait) -- node[anchor=west] {yes} (s);
		
		\draw [arrow] (s) -- node[anchor=south west] {no} (bad);
		\draw [arrow] (s) -- node[anchor=west] {yes} (triangle);
		
		\draw [arrow] (triangle) -- node[anchor=north east] {yes} (bad);
		\draw [arrow] (triangle) -- node[anchor=south] {no} (good);
		
	\end{tikzpicture}
	\hspace{0.5cm}
    \includegraphics[scale=0.55]{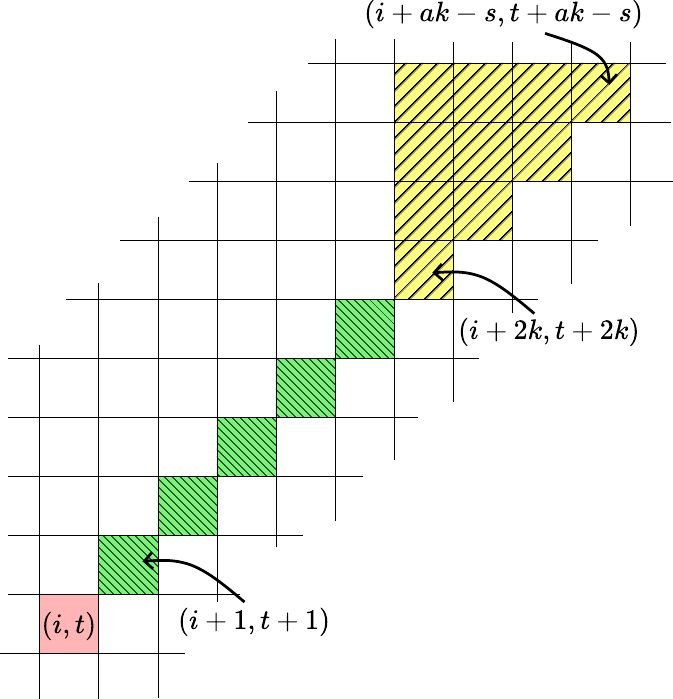}
    \caption{Left: the decision process. Right: its illustration. The green cells are at position $(i+m,t+m)$ with $0 < m < 2k$. The size of the yellow triangle depends on the nature $s$ of the arrow $\nearrow_s$ seen at $(i,t)$.}
    \label{fig:GoodBadOther}
\end{figure}
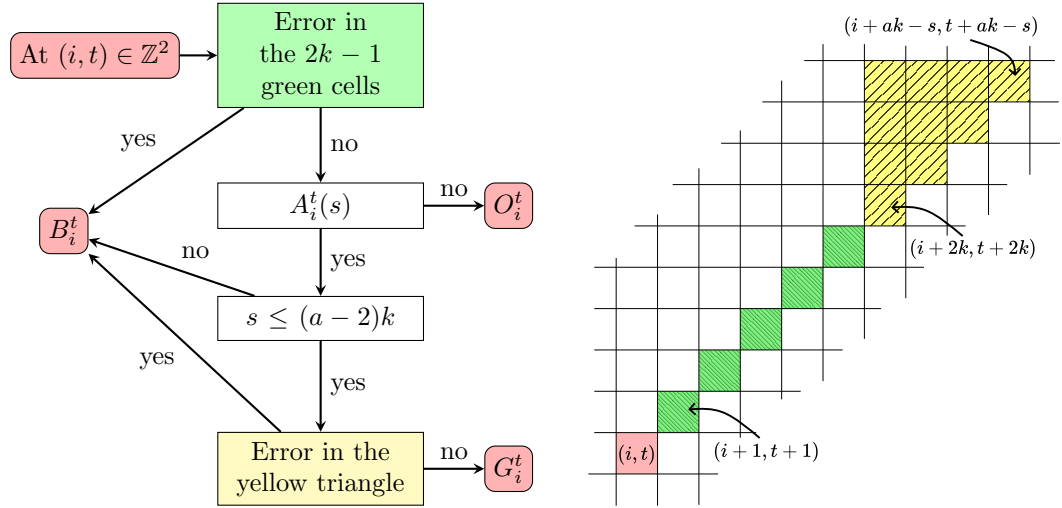

One more formally defines the events as follows.
\begin{itemize}
	\item The good case $G_i^t$, when there is a young arrow ($s\leq (a-2)k$) that defines a large synchronized zone, defined by
	\[
		G_i^t = \bigcup_{s=0}^{(a-2)k} \left[A_i^t(s) \cap \bigcap_{\substack{1 \leq j\leq (a-2)k-s\\ 1 \leq n\leq (a-2)k-s \\ n-j \geq t-i}} \left( E_{i+2k+j}^{t+2k+n} = \bot \right) \right]  \cap \bigcap_{m=1}^{2k-1} \left( E_{i+m}^{t+m} = \bot \right).
	\]
	The probability of this event is 
    \[
    		p_{good}^{\epsilon} \coloneqq \sum_{s=0}^{(a-2)k}P\left(A_i^t(s) \right) (1-\epsilon)^{2k-1+\frac{((a-2)k-s)^2}{2}} \underset{\epsilon\to0}{\sim} \frac{((a-2)k+1)((a-2)k+2)}{2(a+1)k}\epsilon \underset{a\to\infty}{\sim} \frac{ak\epsilon}{2}.
    \]
	\item The bad case $B_i^t$, when there is an arrow with a type too large or a error on its path:
	\[
		B_i^t = \left( \bigcup_{s=(a-2)k+1}^{ak-1} A_i^t(s) \right) \cup \bigcup_{s=0}^{(a-2)k} \left[A_i^t(s) \cap \bigcup_{\substack{1 \leq j\leq (a-2)k-s\\ 1 \leq n\leq (a-2)k-s \\ n-j \geq t-i}} \left( E_{i+2k+j}^{t+2k+n} \neq \bot \right)  \right] \cup \bigcup_{m=1}^{2k-1} \left( E_{i+m}^{t+m} \neq \bot \right).
	\]
	The probability of this event is
	\begin{align*}
    	p_{bad}^{\epsilon} &\coloneqq 1-\left(1-\epsilon\right)^{2k-1}\left(1-\sum_{s=(a-2)k+1}^{ak-1}P\left(A_i^t(s) \right) - \sum_{s=0}^{(a-2)k}P\left(A_i^t(s) \right)\left(1-(1-\epsilon)^{\frac{((a-2)k-s)^2}{2}} \right) \right) \\
    	&\underset{\epsilon\to0}{\sim} \left(2k-1+\frac{ak(ak+1) - ((a-2)k+1)((a-2)k+2)}{2(a+1)k}\right)\epsilon \underset{a\to\infty}{\sim} (4k-3)\epsilon.
    \end{align*}
	\item The other case $O_i^t = \overline{G_i^t \cup B_i^t}$, when there is no arrow and no error in the path:
	\[
		O_i^t = \bigcap_{s=0}^{ak-1} \overline{A_i^t(s)} \cap \bigcap_{m=1}^{2k-1} \left( E_{i+m}^{t+m} = \bot \right).
	\]
	Its probability is then
    \[
    p_{other}^{\epsilon}\coloneqq1-p_{bad}^{\epsilon}-p_{good}^{\epsilon}=1-s(\epsilon)
    \]
    with $s(\epsilon)\tendsto{\epsilon}00$.
  \end{itemize}

  The following property is crucial.

  \begin{lem}
    \label{lem:independent-good-bad-other}
    The event $O_i^t$ is independent of $O_i^{t'}$, $G_i^{t'}$ and $B_i^{t'}$ for all $t' \neq t$.
  \end{lem}
  
  \begin{proof}
  The event $O_i^t$ only depends on the values $E_{i+m}^{t+m}$ for $-ak \leq m \leq 2k-1$, while the other events depend on the values $E_{i+m}^{t'+m}$ for $-ak \leq m \leq 2k-1$ and $E_{i+2k+j}^{i+2k+n}$ for some $j,n \geq 1$.
  As these are disjoint sets of independent random variables, the events are independent.
  \end{proof}

\begin{rem}
The ratio $\frac{p_{good}^{\epsilon}}{s(\epsilon)}=\frac{p_{good}^{\epsilon}}{p_{good}^{\epsilon}+p_{bad}^{\epsilon}}$ often appear in the computations. Using the definitions of $p_{good}^\epsilon$ and $p_{bad}^\epsilon$, one can observe that when $\epsilon\to0$, the ratio has a limit $C_{a,k}$ which verifies: 
\begin{equation}
C_{a,k} \underset{a\to\infty}{=}1-\frac{8-\frac{6}{k}}{a}+o\left(\frac{1}{a}\right).
\label{eq:ratio}
\end{equation}
\end{rem}

\begin{figure}[h]
\centering{}\includegraphics[scale=0.6]{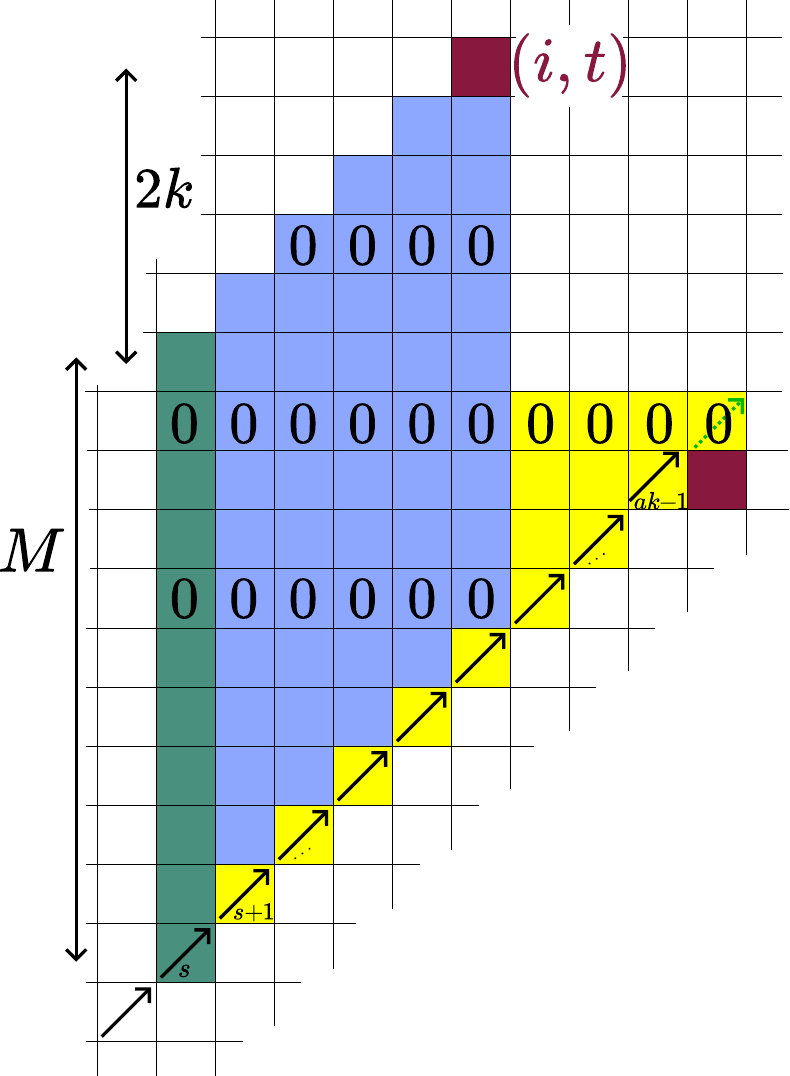} \caption{Suppose $L_{n}=i$ and $t_{n}=t$. In green, the column where the first arrow is looked for: it appears after $M-1$ ``other" cases (no errors on the $2k$ cells in diagonal, in blue). As the arrow is of type $s$ and there is no error on its path (yellow cells), it's a success: the arrow has created a synchronized zone of size greter than $2k$, which does not let information go through from the left. We then define $L_{n+1}=i+(a-2)k-s$ and $t_{n+1}=t-M+\left(a-2\right)k-s-1$.\label{fig:LeftBorder}}
\end{figure}

Figure \ref{fig:LeftBorder} illustrates a transition of the MAC $(L_{n},J_{n})$ which is defined with a random sequence of times $(t_{n})$ by induction, the base case being $L_{0}=0$, $J_{0}=1$, $t_{0}=0$.
Suppose that $L_{n}$, $J_{n}$ and $t_{n}$ have been defined.
If $J_n = 1$ (meaning that the chain has not recently moved to the left), then let $M \coloneqq \min \left\{ 0 \leq m \leq H \mid O_{L_n-2k}^{t_n-2k-m} \text{ does not occur}\right\}$ with $H = H(\epsilon,a)$ a \define{barrier}, to be defined later.
If the set is empty or $M < (a-2)k$ or $B_{L_n-2k}^{t_n-2k-M}$ occurs, we have a \textbf{failure}, and the chain moves to the left:
\begin{equation}
  \label{eq:Ln-failure}
  L_{n+1} = L_n - 4ak, \: J_{n+1} = 2 \text{ and } t_{n+1} = t_n-4ak.
\end{equation}
Otherwise, $G_{L_n-2k}^{t_n-2k-M}$ occurs with the arrow $A_{L_n-2k}^{t_n-2k-M}(s)$ for some $s \leq (a-2)k$, and we have a \textbf{success}.
The chain now moves to the right, onto the position where the arrow disappears naturally:
\[
  L_{n+1} = L_n + (a-2)k-s, \: J_{n+1} = 1 \text{ and } t_{n+1} = t_n - M + (a-2)k -s -1.
\]
Denote the probability of \textbf{success} by $\alpha_\epsilon$.

In the case $J_n = 2$, the process is a little more involved.
As in the first case, define $M_1 \coloneqq \min\left\{ 0 \leq m \leq H \mid O_{L_n-2k}^{t_n-2k-m} \text{ does not occur}\right\}$.
If we have a \textbf{failure}, then we use Equation \eqref{eq:Ln-failure} as before.

Suppose then that we have a \textbf{success} with some $s \leq (a-2)k$.
If $M_1 - (a-2)k +s +1 < H$, then the barrier has not been overcome and we call it a \textbf{jump}.
We then fix a new $M_2$ and iterate the last step.
More precisely, if the first $l > 0$ attempts have resulted in \textbf{jumps}, then let $M_{l+1} \coloneqq \min\left\{ m \leq H \mid O_{L_n-2k}^{t_n-2k-m-\sum_{r=1}^{l}M_r } \text{ does not occur}\right\}$.
If the set is empty or $B_{L_n-2k}^{t_n-2k-\sum_{r=1}^{l+1}M_r}$ occurs, we have a \textbf{failure} and apply Equation \eqref{eq:Ln-failure}.
Otherwise, if $\sum_{r=1}^{l}M_r - (a-2)k +s_l +1 < H$ (with $s_l$ the $s$ of the last success) we have a \textbf{jump}.
If we have neither a \textbf{failure} nor a \textbf{jump}, then the barrier $H$ has been overcome and the chain moves to the right:
\[
  L_{n+1} = L_n + (a-2)k-s_l, \: J_{n+1} = 1 \text{ and } t_{n+1} = t_n - \sum_{r=1}^{l}M_r + (a-2)k -s_l -1.
\]
Denote the probability of overcoming a barrier by $\beta_\epsilon$.

A transition of $L_{n}$ with $J_{n}=1$ (a check) is illustrated on Figure \ref{fig:LeftBorder}, and a trajectory of $\left(L_{n},J_{n}\right)$ on Figure \ref{fig:MAC}. A success is obtained when the arrow defining the synchronized zone the considered cell is part of is encountered, and this zone is sufficiently large. Otherwise, it is a failure. The barrier $H$ has no physical reality in the trajectory: it is only defined to ensure the mutual independence between the different jumps. Aiming for a large probability of success (encountering a good case before $H$ tries) while also having a great chance to overcome a barrier of size $H$, a good choice for $H$ is $H\coloneqq\frac{\ln(a)}{s(\epsilon)}$ (where $s\left(\epsilon\right)=1-p_{other}^{\epsilon}$).

\begin{figure}[!h]
\centering{}\includegraphics[scale=0.5]{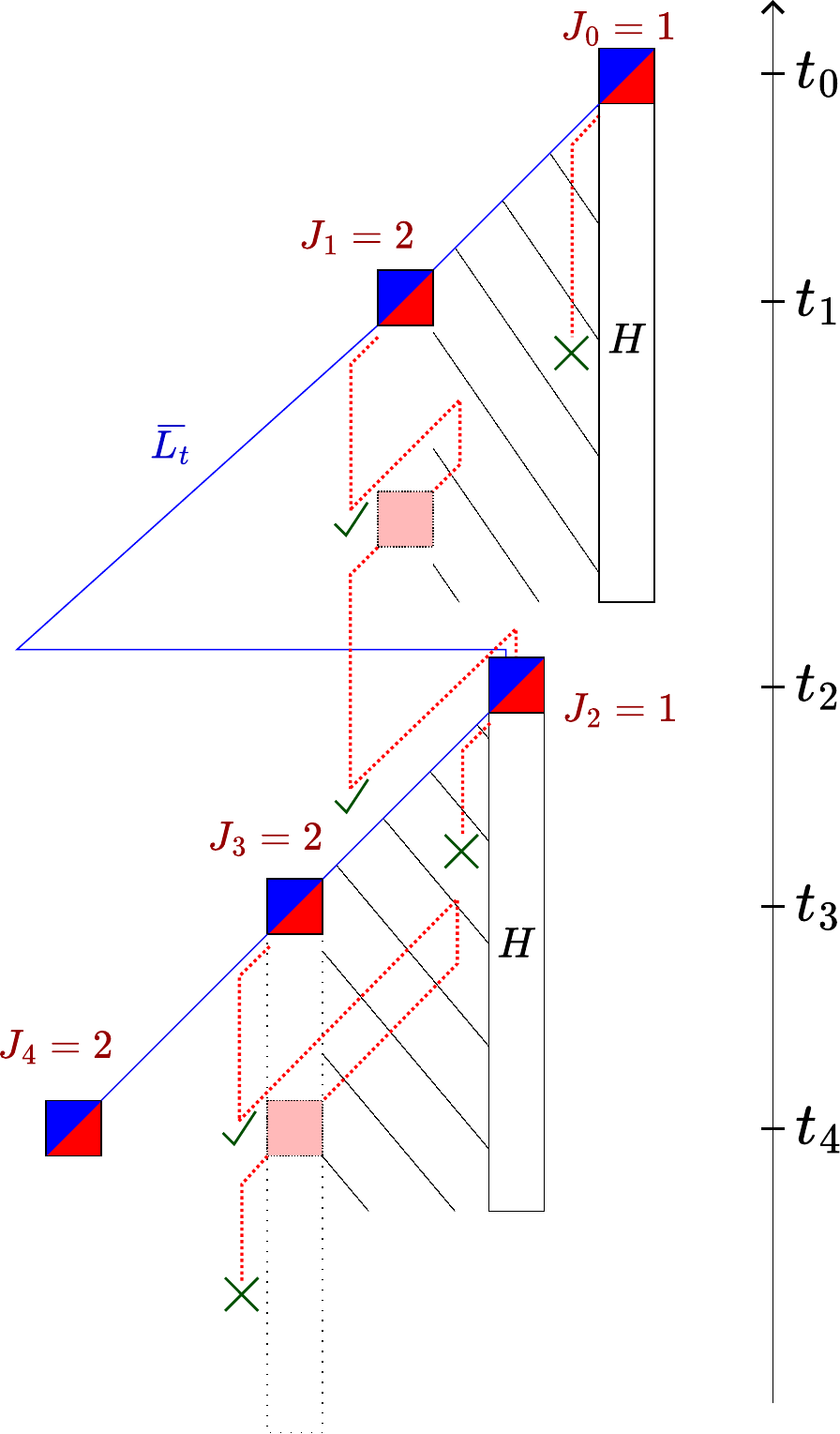} \caption{First steps of $(L_{n},J_{n})$ and $\overline{L_{t}}$ (in blue). The red-and-blue cells are at times $t_{n}$ in position $L_{n}$. The initial state is $(L_{0},J_{0})=(0,1)$. The first steps represented are: \protect \\
(1) The initial cell is not in a good synchronized zone: we encounter a bad case before $H$ checks, it's a failure. We represent a barrier of size $H$ and the new state is $(L_{1},2)$ with $L_{1}<L_{0}$.\protect \\
(2) The first result is a success (the synchronized zone is large enough), but before $H$ checks: it's a jump. The second one succeed after enough checks to overcome the barrier: the last success defines the new state $L_{2}\protect\geq L_{1}$ and $J_{2}=1$.\protect \\
(3) The check fails, so $L_{3}<L_{2}$ and a barrier is represented.\protect \\
(4) The first check is a jump, but the second is a failure: $L_{4}<L_{3}$ and a new barrier is represented. \label{fig:MAC}}
\end{figure}

The sequence $(J_{n})$ is an ergodic Markov chain on $\{1,2\}$ with unique invariant probability vector $\pi=(\pi_{1},\pi_{2})\coloneqq\frac{1}{1-\alpha_{\epsilon}+\beta_{\epsilon}}(\beta_{\epsilon},1-\alpha_{\epsilon})$. Moreover, the level change $L_{1}-L_{0}$ knowing $J_{1}$ is independent from $J_{0}$: it equals $-4ak$ if $J_{1}=2$, and $(a-2)k-s$ if $J_{1}=1$, with $s$ independent from $J_{0}$.

\subsection{Interpretation of $L_n$}

The sequence of spacetime positions $\left(L_{n},t_{n}\right)$ is used to represent points such that the information to its left on the $\G$-layer cannot affect the right half-configuration at $t=0$. After a failure, the sequence moves to the left with slope $1$, the maximum speed of information propagation in $T_\epsilon$ (Lemma \ref{lem:passageOfTime}). After a success, $L_{n}$ is at the border of a synchronized zone of size greater than $2k$, through which information cannot flow.
This is formalized in the following result.

\begin{prop}
  \label{prop:Ln-lt}
  For all $\omega\in\Omega$ and  $n\in\N$, 
  \[
    l_{-t_n+2ak} \geq L_n - 2ak.
  \]
\end{prop}

\begin{proof}
  For $s\geq -t_n$ define 
  \[
    l_{s}^{n}\coloneqq
    \min\left\{ m\geq L_{n}-s-t_{n}\mid\exists\beta\in\A^{\N},\,
      \begin{array}{ccl}
        \A^{h\left(m,n\right)} & \to & \A^{\N}\\
        \alpha & \mapsto & \Psi_{L_{n}^{+}}^{s+t_{n}}\left(\alpha\beta,U^{+t_{n}}\right)
      \end{array}
      \text{ is not constant}\right\} 
  \]
  where $h\left(m,n\right)=m+s+t_{n}-L_{n}+1$ and $\left(U^{+t_{n}}\right)^{t}=U^{t+t_{n}}$. It is the left border of the dependence cone at time $-s$ of the semi-configuration on $[L_n, +\infty)$ at time $t_n$.

  For a given $\omega \in \Omega$, let us prove the following proposition by induction on $n\in\N$:
  \begin{equation}
    \label{eq:ln-lnt-bound}
    \forall t \leq t_n -2ak, \quad l_{-t} \geq l^n_{-t}
  \end{equation}
  Intuitively, this states that any information at time $t \leq t_n -2ak$ that cannot influence the half-configuration on $[L_n, \infty)$ at time $t$, also cannot influence the half-configuration on $[0, \infty)$ at time $0$.
  This implies the original claim, as information can propagate at speed at most $1$ (Lemma \ref{lem:passageOfTime}).

  The base case $n=0$ is immediate by definition, $l_{t} = l^0_{t}$.
  Suppose now that \eqref{eq:ln-lnt-bound} is true for a fixed $n\in\N$. By $t_{n+1} < t_n$, the inequality $l_{-t} \geq l^n_{-t}$ is true for all $t\leq t_{n+1}-2ak$.

  If $J_{n+1} = 2$, the last step is a failure: $L_{n+1} = L_n -4ak$ and $t_{n+1} = t_n-4ak$. By Lemma \ref{lem:passageOfTime}, 
  \[
    l^n_{t_{n+1}} \geq L_n -4ak = L_{n+1}.
  \]
  Therefore the value of the configuration on $[L_n, +\infty)$ at $t_n$ is determined by the value of the configuration on $[L_{n+1}, +\infty)$ at $t_{n+1}$, which gives by definition $l^n_{-t} \geq l^{n+1}_{-t}$ for $t\leq t_{n+1}-2ak$ and the result.

  If $J_{n+1} = 1$, the last step is a success: the cell at $L_n$ is at the right end of a zone of size greater than $2k$ of synchronized values on the $\F$-layer.
  This zone is created by an arrow that dies at the position $(J_{n+1}, t_{n+1})$.
  Any information on the $\G$-layer that is to the left of the arrow at time $t_{n+1}$ cannot influence the value of the configuration on $[L_n; +\infty)$ at time $t_n$, as it would have to pass through a region whose cells are forced to take a value $0_{\G}$ every $k$ steps, the last of which happens at time $t_{n+1}+1$.
  Likewise, no arrows on the $\F$-layer cross this region.
  Hence the value of the configuration on $[L_n, +\infty)$ at $t_n$ is determined by the value of the configuration on $[L_{n+1}, +\infty)$ at $t_{n+1}$, which gives the result. Figure \ref{fig:Ln-lt} illustrates the right shift of the left border of the dependence cone after a success. 
\end{proof}

\begin{figure}[!h]
\centering{}
\begin{tikzpicture}[scale=0.5]

  \pgfmathsetmacro{\height}{10}

  \draw[fill=red!30] (0,0) rectangle ++(1,\height+8);
  
  \draw[fill=green!30] (1,1)
  \foreach \x in {1,...,4}{
    -| ++(1,1)
  }
  |- ++(-2,\height+6)
  \foreach \x in {1,...,2}{
    |- ++(-1,-1)
  }
  -- cycle;

  \draw[fill=yellow!30] (5,5)
  \foreach \x in {1,...,6}{
    -| ++(1,1)
  }
  -| cycle;

  \node [anchor=west] at (0.25,0.25) {${\scriptstyle s}$};
  \foreach \x in {-2,...,10}{
    \node at (\x+0.5,\x+0.5) {$\nearrow$};
  }

  \foreach \x in {0,...,3}{
    \node at (\x+0.5,3.5) {$0$};
  }
  \foreach \x in {0,...,7}{
    \node at (\x+0.5,7.5) {$0$};
  }
  \foreach \x in {0,...,4}{
    \node at (\x+0.5,11.5) {$0$};
  }
  \foreach \x in {0,...,4}{
    \node at (\x+0.5,15.5) {$0$};
  }
  \foreach \x in {2,3,4}{
    \node at (\x+0.5,19.5) {$0$};
  }

  \draw (14,-5.5) -- (14,20.5);
  \draw (13.75,-5.5) -- ++(0.5,0);
  \node [right] at (14.5,-5.5) {$t_{n+1}-2ak$};
  \draw (13.75,10.5) -- ++(0.5,0);
  \node [right] at (14.5,10.5) {$t_{n+1}$};
  \draw (13.75,20.5) -- ++(0.5,0);
  \node [right] at (14.5,20.5) {$t_n$};
  
  \draw [thick, dotted, ->] (-4.5,-5.5) -- ++(16.5,16.5);
  \draw [thick, dotted, ->] (-6.5,-5.5) -- ++(8.5,8.5);
  \draw [thick, dotted, ->] (-6.5,1.5) -- ++(9.5,9.5);
  \draw [thick, dotted, ->] (-6.5,6.5) -- ++(8.5,8.5);

  \fill (4.5,20.5) circle (0.15cm);
  \node [left] at (4.5,20.5) {$L_n$};
  \draw[dashed] (4.5,20.5) -- (13.75,20.5);

  \fill (11.5,10.5) circle (0.15cm);
  \node [below right] at (11.5,10.5) {$L_{n+1}$};
  \draw[dashed] (11.5,10.5) -- (13.75,10.5);
  
  \fill (-4.5,-5.5) circle (0.15cm);
  \draw[dashed] (-4.5,-5.5) -- (13.75,-5.5);
  
\end{tikzpicture}
\caption{Illustration of a success in $L_{n}$ (not to scale). The colored area is synchronized and blocks all information flow on the $\G$-layer. The colors correspond roughly to Figure \ref{fig:GoodBadOther}; the red strip consists of ``other'' events followed by a ``good'' event (with $\nearrow_s$). In dots the theoretical information flow at speed $1$. For the configuration at $t_{n+1}-2ak$, only the information to the right of the dot can influence the half-configuration right of $(L_n, t_n)$, passing to the right of the synchronized zone. They are included in the cells that can influence the half-configuration right of $(L_{n+1}, t_{n+1})$. \label{fig:Ln-lt}}
\end{figure}
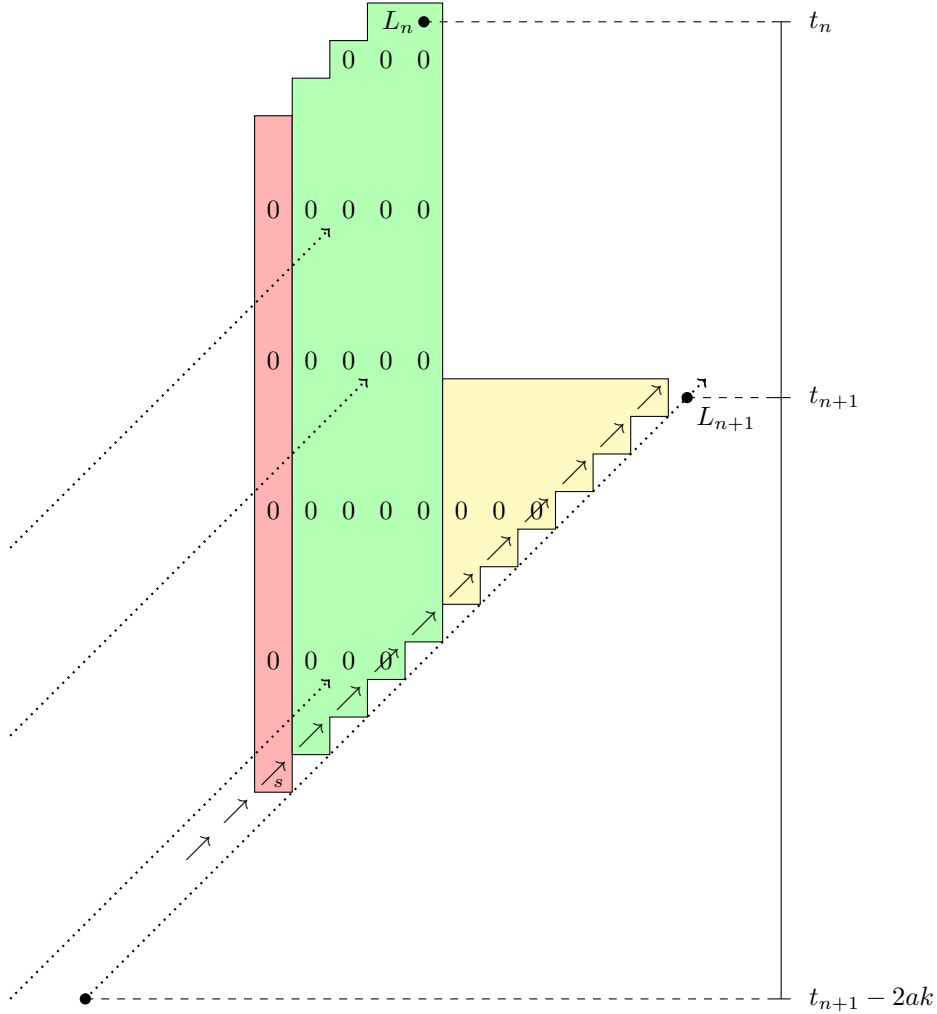

To more easily represent the bound of $l_t$ by $L_n$, we define its ``hull'' as $\overline{L_{t}} \coloneqq L_n - (t_n - t)$, where $n$ is such that $t_{n+1} < t \leq t_n$. The following corollary is a direct application of the last proposition and Lemma \ref{lem:passageOfTime}:
\begin{cor}
For a given $\omega\in\Omega$ we have $l_{-t+2ak} \geq \overline{L_t} - 2ak$.
\end{cor}

\begin{cor}
For a given $\omega\in\Omega$, if $L_n \tendsto{n}{\+\infty}+\infty$, then $l_t \tendsto{t}{+\infty}+\infty$.
\end{cor}
\begin{proof}
By construction, $(t_n)$ is a decreasing sequence such that $t_n - t_{n+1} \leq 2H$ for all $n$. Therefore for any $t_{n+1} < t \leq t_n$, we have $\overline{L_t} \geq L_n - 2H$
and $\overline{L_t} \tendsto{t}{-\infty}+\infty$. The previous corollary then gives the result.
\end{proof}

\subsection{Mean drift of the left border}

By the properties verified by the MAC $(L_n,J_n)$, its mean drift is
\begin{align*}
	\kappa^\prime &= \pi_1 \sum_{m\in\Z} m P\left(L_1 = m \mid J_1=1\right) + \pi_2 \sum_{m\in\Z} m P\left(L_1 = m \mid J_1=2\right)\\
		&= \pi_1 \sum_{s=0}^{(a-2)k}((a-2)k-s)  P\left(L_1 = ((a-2)k-s) \mid J_1=1\right) - \pi_{2}4ak
\end{align*}
Knowing that $J_1=1$, the step to the right is $(a-2)k-s$ if $\nearrow_s$ is the arrow involved in the last success (independently of $J_0 = 1$ or $2$). As it is a success, the arrow can only be $\nearrow_s$ with $s \leq (a-2)k$, and the probability is independent of the actual cell the event $A_i^t(s)$ occurred: 
\[
	P\left(L_1 = ((a-2)k-s) \mid J_1=1\right) = \frac{P\left(A_0^0(s)\right)}{P\left(\bigcup_{l=0}^{(a-2)k}A_0^0(l)\right)}.
\]

Therefore the mean drift is strictly positive if $\beta_{\epsilon}$ and $\alpha_{\epsilon}$ verify the following inequality
\[
	\pi_{1} \sum_{s=0}^{(a-2)k}(ak-s-2k) \frac{P\left(A_0^0(s)\right)}{P\left(\bigcup_{l=0}^{(a-2)k}A_0^0(l)\right)} > \pi_{2}4ak
\]
which becomes when we multiply both side by $1-\alpha_\epsilon + \beta_\epsilon$:
\begin{equation}
\underset{\text{Section \ref{sec:beta}}}{\underbrace{\beta_{\epsilon}}} \underset{\text{Section \ref{sec:rightStep}}}{\underbrace{\sum_{s=0}^{(a-2)k}(ak-s-2k)\frac{P\left(A_0^0(s)\right)}{P\left(\bigcup_{l=0}^{(a-2)k}A_0^0(l)\right)}}} > \underset{\text{Section \ref{sec:alpha}}}{\underbrace{\left(1-\alpha_{\epsilon}\right)}}4ak.
\label{eq:driftLeftBorder}
\end{equation}

\subsection{Mean step to the right}
\label{sec:rightStep}

We already computed $P\left(A_{i}^{t}(s)\right)=\sum_{m=0}^{s}\frac{\epsilon}{(a+1)k}(1-\epsilon)^{m}\underset{\epsilon\to0}{\sim}\frac{s+1}{(a+1)k}\epsilon$ for all $\left(i,t\right)$. As the $A_{0}^{0}\left(l\right)$ are disjoints for $l\in\left\llbracket 0,\left(a-2\right)k\right\rrbracket $, we can show that
\[
	\frac{P\left(A_{0}^{0}(s)\right)}{P\left(\bigcup_{l=0}^{(a-2)k}A_{0}^{0}(l)\right)}\tendsto{\epsilon}0\frac{\left(s+1\right)}{\sum_{l=0}^{\left(a-2\right)k}\left(l+1\right)}=\frac{2\left(s+1\right)}{\left(\left(a-2\right)k+1\right)\left(\left(a-2\right)k+2\right)}
\]
and so 
\[
	\sum_{s=0}^{(a-2)k}(ak-s-2k)  \frac{P\left(A_0^0(s)\right)}{P\left(\bigcup_{l=0}^{(a-2)k}A_0^0(l)\right)}  \tendsto{\epsilon}{0} \frac{(a-2)k}{3}.
\]

Thus inequality (\ref{eq:driftLeftBorder}) is verified for $\epsilon$ small enough if 
\[
	\liminf_{\epsilon\to0}\beta_{\epsilon}>\frac{12a}{a-2}\limsup_{\epsilon\to0}\left(1-\alpha_{\epsilon}\right).
\]

\subsection{Computation of $\alpha_{\epsilon}$, the probability of success}
\label{sec:alpha}

Recall that the number $\alpha_{\epsilon}$ is the probability of success when determining the direction of the next step of $L_{n}$ with $J_{n} = 1$ (see Section \ref{subsec:L_n}), i.e. seeing a good case before a bad one and between $ak$ and $H$ timesteps.
By Lemma \ref{lem:independent-good-bad-other}, we have
\[
	\alpha_{\epsilon} = \sum_{n=ak}^{H}(p_{other}^{\epsilon})^{n-2k}p_{good}^\epsilon = \frac{p_{good}^{\epsilon}}{1-p_{other}^{\epsilon}}\left(1-(p_{other}^{\epsilon})^{H-2k}\right)(p_{other}^\epsilon)^{(a-2)k}
\]
As $H=\frac{\ln(a)}{s(\epsilon)}$, one has $(p_{other}^{\epsilon})^{H}=\exp\left(\ln(a)\frac{\ln(1-s(\epsilon))}{s(\epsilon)}\right)\tendsto{\epsilon}0\frac{1}{a}$, $(p_{other}^{\epsilon})^{-2k} \tendsto{\epsilon}{0}1$, $(p_{other}^{\epsilon})^{(a-2)k} \tendsto{\epsilon}{0}1$ and
\[
	\alpha_{\epsilon}\tendsto{\epsilon}0 C_{a,k} \left(1-\frac{1}{a}\right)\eqqcolon\alpha
\]

Thus inequality (\ref{eq:driftLeftBorder}) is verified for $\epsilon$ small enough if $\liminf_{\epsilon\to0}\beta_{\epsilon}>\frac{12a}{a-2}(1-\alpha).$
Moreover, by equation (\ref{eq:ratio}) we have the expansion $\alpha\underset{a\to\infty}{=}1-\frac{9-\frac{6}{k}}{a}+o\left(\frac{1}{a}\right)$ so to conclude one only needs that for $a$ large enough
\[
	\liminf_{\epsilon\to0}\beta_{\epsilon} \geq \frac{109}{a}.
\]

\subsection{Distribution of the jumps}

Suppose that the MAC is in state $J_n = 2$ at step $n > 0$, so that there is a barrier of height $H$ on the right.
Now we iterate the check of the $J_n = 1$ case.
Each check has a probability $\alpha_\epsilon$ of \textbf{success}, that is, seeing a ``good'' case $G_i^t$ before a ``bad'' case $B_i^t$.
The \define{size} of the \textbf{success} is the number of ``other'' cases $O_i^t$ before the first ``good'' case.
We keep repeating the check until either a \textbf{failure} occurs or the sum of the sizes of the \textbf{success}es so far is more than $H$; before that, each \textbf{success} is a \textbf{jump}.
If Figure \ref{fig:LeftBorder} depicted a \textbf{jump}, the new cell would be at position $(i,t-n+2k)$, so its size would be $n-2k$.

We now analyze the size distribution of a \textbf{success}.
Denote by $N$ the size of a \textbf{success}.
Supposing that a \textbf{success} actually occurs, we find for $(a-2)k\leq d \leq H-2k$ that
\[
	P\left(N=d\mid\text{success}\right) = \frac{(p_{other}^{\epsilon})^{d} \, p_{good}^{\epsilon}}{\sum_{l=ak}^{H}(p_{other}^{\epsilon})^{l-2k} \, p_{good}^{\epsilon}} = (p_{other}^{\epsilon})^{d-(a-2)k}  \frac{s(\epsilon)}{1-(p_{other}^{\epsilon})^{H-ak}}
\]
and thus for $\frac{H}{i}\geq (a-2)k$,
\begin{align*}
P\left(N\geq\frac{H}{i}\mid\text{success}\right) & =\sum_{d=\frac{H}{i}}^{H}P\left(N=d\mid\text{success}\right)\\
	& = \frac{s(\epsilon)}{1-(p_{other}^{\epsilon})^{H-ak}} \sum_{d=\frac{H}{i}}^{H}(p_{other}^{\epsilon})^{d-(a-2)k}\\
	& = (p_{other}^{\epsilon})^{\frac{H}{i}-(a-2)k} \frac{1-(p_{other}^{\epsilon})^{\frac{i-1}{i}H}}{1-(p_{other}^{\epsilon})^{H-ak}}
\end{align*}

Using the previous computation of $(p_{other}^{\epsilon})^{H}$, we can deduce that 
\[
P\left(N\geq\frac{H}{i}\mid\text{success}\right) \tendsto{\epsilon}0 \left(\frac{1}{a}\right)^{1/i} \frac{1-\left(\frac{1}{a}\right)^{\frac{i-1}{i}}}{1-\frac{1}{a}}.
\]

\subsection{Computation of $\beta_{\epsilon}$, the probability of overcoming a barrier}
\label{sec:beta}

Recall that $\beta_{\epsilon}$ is the probability of overcoming a barrier. Let us fix $m \geq 1$.
For $i \geq 1$, let $A_{i}$ be the event that we have $i$ \textbf{success}es in a row and each has size at least $H/i$.
If $A_{i}$ occurs, then we realize $i$ \textbf{success}es in a row, which are all large enough so that added together, they overcome a barrier $H$.

Let $N_{l}$ be the size of the $l^{th}$ jump.
The random variables $(N_l)_{1 \leq l \leq i}$ are independent and identically distributed given $i$ \textbf{success}es in a row. Finally, define $\mathbf{A}_m = \bigcup_{i=1}^{m}A_{i}$. Then
\begin{align}
\beta_{\epsilon} & \geq P\left(\text{overcoming the barrier in at most }m\text{ jumps}\right) \notag \\
	& \geq P(\mathbf{A}_m) \notag \\
	& =\sum_{n=1}^{m}\left(-1\right)^{n+1}\sum_{1\leq i_{1}<\dots<i_{n}\leq m}P\left(A_{i_{1}}\cap\dots\cap A_{i_{n}}\right).\label{eq:beta_epsilon}
\end{align}
With the convention $i_{0}=0$, we have
\begin{align*}
P\left(A_{i_{1}}\cap\dots\cap A_{i_{n}}\right) & = P\left(\left(\bigcap_{j=1}^{n}\bigcap_{l=i_{j-1}+1}^{i_{j}}N_{l}\geq\frac{H}{i_{j}}\right) \cap i_{n}\text{ \textbf{success}es in a row}\right)\\
	& = \left(\prod_{j=1}^{n}\prod_{l=i_{j-1}+1}^{i_{j}}P\left(N\geq\frac{H}{i_{j}}\mid\text{\textbf{success}}\right)\right) \alpha_{\epsilon}^{i_{n}}\\
	& =\prod_{j=1}^{n}P\left(N\geq\frac{H}{i_{j}}\mid\text{\textbf{success}}\right)^{i_{j}-i_{j-1}} \alpha_{\epsilon}^{i_{n}}.
\end{align*}

Finally, using the computation of the previous section (supposing that $m\leq \frac{H}{ak}$, which is true for $\epsilon$ close enough to $0$) gives
\[
	P\left(A_{i_{1}}\cap\dots\cap A_{i_{n}}\right) \tendsto{\epsilon}0 \prod_{j=1}^{n} \left(\frac{1}{a}\right)^{\frac{i_{j}-i_{j-1}}{i_{j}}}  \underset{\tendsto a{\infty}1}{\underbrace{\left(\frac{1-\left(\frac{1}{a}\right)^{\frac{i_j-1}{i_j}}}{1-\frac{1}{a}} \right)^{i_{j}-i_{j-1}}}} \underset{\tendsto a{\infty}1}{\underbrace{\alpha^{i_{n}}}}.
\]
So when $n\geq2$, this product is equivalent to 
\[
	\prod_{j=1}^{n} \left(\frac{1}{a}\right)^{\frac{i_{j}-i_{j-1}}{i_{j}}} = \left(\frac{1}{a}\right)^{1 + \sum_{j=2}^n \frac{i_{j}-i_{j-1}}{i_{j}}} 
\] 
and so negligible with respect to $\frac{1}{a}$, while in the case $n=1$, $P(A_{i})\underset{a\to\infty}{\sim}\frac{1}{a}$. 

Thus, the limit when $\epsilon\to0$ of the sum in (\ref{eq:beta_epsilon}) is equivalent to $\frac{m}{a}$ when $a\to +\infty$. Taking $m$ greater than $109$ and $a$ big enough, the following holds:
\begin{equation}
\liminf_{\epsilon\to0}\beta_{\epsilon}\geq \lim_{\epsilon\to0} P(\mathbf{A}_m) \geq \frac{109}{a} \label{eq:conclusionDerive1}
\end{equation}
Thus, equation (\ref{eq:driftLeftBorder}) is verified: the level $L_n$ tends to $+\infty$ a.s..

\subsection{The right border}
\label{subsec:Bordure-Droite} 

We can use analogous definitions as for the left border to define a MAC $(R_n,K_n)$ on $\Z\times\{1,2\}$ bounding $r_t$, with the same ideas of synchronized zones and barriers. This time, we prove that the mean drift $\kappa^\prime$ is strictly negative to deduce that $R_n \tendsto{n}{\infty}-\infty$.
We skip some of the details that are analogous to the case of $(L_n, J_n)$, and reuse some of the notation with possibly different definitions.

\subsubsection{Definition of $R_n$}

The error variables $E_{i}^{t}$ are the same as in the definition of $L_{n}$. For the right border, the type $s$ of the arrow $\nearrow_s$ is less important than its age $l$ (the number of iterations since its creation by an error). For $0\leq l<ak$, the event $F_{i}^{t}\left(l\right)$ ``having an arrow with age $l$ at $\left(i,t\right)$" is defined by
\[
	F_{i}^{t}\left(l\right)\coloneqq\left(\bigcup_{s=0}^{ak-l-1}\left(E_{i-l}^{t-l}=\nearrow_{s}\right)\right)\cap\left(\bigcap_{n=0}^{l-1}\left(E_{i-n}^{t-n}=\bot\right)\right)
\]
with probability
\[
	P\left(F_{i}^{t}\left(l\right)\right)=\sum_{s=0}^{ak-l-1}\frac{\epsilon}{\left(a+1\right)k}\left(1-\epsilon\right)^{l}=\frac{ak-l}{\left(a+1\right)k}\epsilon\left(1-\epsilon\right)^{l}.
\]
We reuse the notations $O_{i}^{t}$, $B_{i}^{t}$ and $G_{i}^{t}$ for the other, bad and good cases, here defined by the decision process of Figure \ref{fig:GoodBadOtherDroite}.

\begin{figure}[!h]
\centering{}\begin{tikzpicture}[node distance=2cm, startstop/.style={rectangle, rounded corners,  text centered,  draw=black,  fill=red!30}, decision/.style={rectangle, text width=2.5cm,  text centered,  draw=black}, arrow/.style={thick,->,>=stealth}]

		\node (start) [startstop] {At $(i,t)\in\Z^2$};
		\node (errorPath) [decision, fill=green!30, right of=start, xshift=1cm] {Error in the $2k+1$ green cells};
		\node (Ait) [decision, below of=errorPath] {$F_i^t(l)$};
		\node (s) [decision, below of=Ait, yshift=0.5cm] {$l > 2k$};
		\node (triangle) [decision, fill=yellow!30, below of=s] {Error in the yellow zone};
		\node (other) [startstop, right of=Ait, xshift = 0.5cm] {$O_i^t$};
		\node (bad) [startstop, above left of=s, xshift = -2cm, yshift=-0.25cm] {$B_i^t$};
		\node (good) [startstop, right of=triangle, xshift = 0.5cm] {$G_i^t$};
		
		\draw [arrow] (start) -- (errorPath);
		
		\draw [arrow] (errorPath) -- node[anchor=west] {no} (Ait);
		\draw [arrow] (errorPath) -- node[anchor=south east] {yes} (bad);
		
		\draw [arrow] (Ait) -- node[anchor=south] {no} (other);
		\draw [arrow] (Ait) -- node[anchor=west] {yes} (s);
		
		\draw [arrow] (s) -- node[anchor=south west] {no} (bad);
		\draw [arrow] (s) -- node[anchor=west] {yes} (triangle);
		
		\draw [arrow] (triangle) -- node[anchor=north east] {yes} (bad);
		\draw [arrow] (triangle) -- node[anchor=south] {no} (good);
\end{tikzpicture}
\hspace{0.5cm}
	 \includegraphics[scale=0.55]{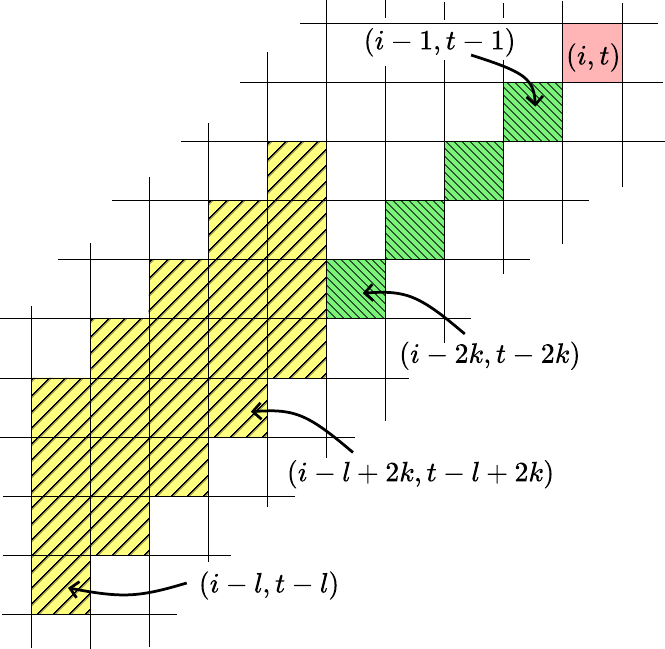} \caption{Left: the decision process. Right: its illustration. The green cells are in position $(i+m,t+m)$ with $0<m<2k$. The yellow zone is of height $2k+1$.\label{fig:GoodBadOtherDroite}}
\end{figure}
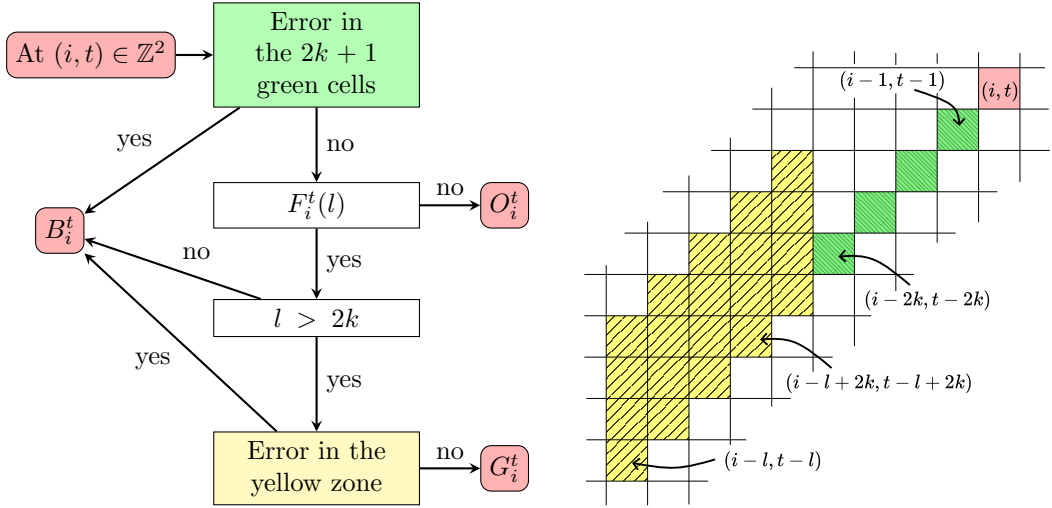

We can then compute the respective probabilities of these events:

\[
	P\left(G_{i}^{t}\right)=p_{good}^{\epsilon}\coloneqq\sum_{l=2k+1}^{ak-1}P\left(F_{i}^{t}\left(l\right)\right)\left(1-\epsilon\right)^{\left(l-2k+1\right)2k}\underset{\epsilon\to0}{\sim}\frac{(a-2)k((a-2)k-1)}{2(a+1)k}\epsilon\underset{a\to\infty}{\sim}\frac{ak\epsilon}{2}
\]
\[
	P\left(B_{i}^{t}\right)=p_{bad}^{\epsilon}\coloneqq1-\left(1-\epsilon\right)^{2k+1}+\sum_{l=2k+1}^{ak-1}P\left(F_{i}^{t}\left(l\right)\right)\left(1-\left(1-\epsilon\right)^{\left(l-2k+1\right)2k}\right)\underset{\epsilon\to0}{\sim}(2k+1)\epsilon.
\]
and $P\left(O_{i}^{t}\right)=p_{other}^{\epsilon}=1-p_{good}^{\epsilon}-p_{bad}^{\epsilon}=1-s\left(\epsilon\right)$ with $s\left(\epsilon\right)\tendsto{\epsilon}00$.
We can also compute the limit
\[
	\frac{p_{good}^{\epsilon}}{p_{good}^{\epsilon}+p_{bad}^{\epsilon}}\tendsto{\epsilon}0 D_{a,k} \underset{a\to\infty}{=}1-\frac{4+\frac{2}{k}}{a}+o\left(\frac{1}{a}\right)
      \]
      
The barrier size is again defined to be $H=\frac{\ln a}{s\left(\epsilon\right)}$. For $\left(i,t\right)\in\Z^{2}$, denote by
\[
	CI\left(i,t\right)=\left\{ \left(i+c,t-l\right)\mid0\leq c<2k,\,c\leq l<4k-c\right\} 
\]
the light blue zone illustrated on Figure \ref{fig:RightBorder}. The MAC $\left(R_{n},K_{n}\right)$ and its associated random time sequence $\left(t_{n}^{\prime}\right)$ is then defined analogously to $\left(L_{n},J_{n}\right)$.
The base case is $R_{0}=0$, $K_{0}=1$, $t_{0}^{\prime}=0$.
For all $n\in\N$, in the case $K_n = 1$ we define as follows.
	\begin{itemize}
		\item Fix $M\coloneqq\min\left\{ 0\leq m\leq H\mid O_{R_{n}+2k}^{t_{n}^{\prime}-2k-m}\text{ does not hold}\right\} $. 
		\item If the set is empty, or $\bigcup_{j,s\in CI\left(R_{n},t_{n}^{\prime}\right)}E_{j}^{s}\neq\bot$, or $B_{R_{n}+2k}^{t_{n}-2k-M}$ holds, we have a \textbf{failure}: 
                  \begin{equation}
                    \label{eq:Rn-failure}
			R_{n+1}=R_{n}+4ak,\:K_{n+1}=2\text{ and }t_{n+1}^{\prime}=t_{n}^{\prime}-4ak.
		\end{equation}
		\item Otherwise, if $G_{R_{n}+2k}^{t_{n}-2k-M}$ holds with $F_{R_{n}+2k}^{t_{n}-2k-M}(l)$, we have a \textbf{success}: 
		\[
			R_{n+1}=R_{n}+4k+1-l,\:K_{n+1}=1\text{ and }t_{n+1}^{\prime}=t_{n}^{\prime}-M-l.
		\]
		Denote by $\alpha_{\epsilon}$ the probability of \textbf{success}.
              \end{itemize}
In the case $K_n = 2$, we define as follows.
	\begin{itemize}
		\item Fix $M_{1}\coloneqq\min\left\{ 0\leq m\leq H\mid O_{R_{n}+2k}^{t_{n}^{\prime}-2k-m}\text{ does not hold}\right\} $ as in the case $K_n = 1$.
		\item If we have a \textbf{failure}, we apply Equation \ref{eq:Rn-failure}.
              \item If we have a \textbf{success} but $M_{1}+l<H$, then it's only a \textbf{jump}.
                Then we fix a new $M_2$ and iterate the last step.
                In general, fix $M_{j+1}\coloneqq\min\left\{ m\leq H\mid O_{R_{n}+2k}^{t_{n}^{\prime}-\left(4j+2\right)k-m-\sum_{r=1}^{j}M_{r}}\text{ is not verified}\right\} $.
                If we have enough \textbf{success}es in a row with $\sum_{r=1}^{l}M_{r}+4jk+l_{j}\geq H$ (with $l_{j}$ the $l$ of the last success), then the barrier $H$ is overcome:
		\[
			R_{n+1}=R_{n}+4k+1-l_{j},\:K_{n+1}=1\text{ and }t_{n+1}^{\prime}=t_{n}^{\prime}-\sum_{r=1}^{l}M_{r}-4jk-l_{j}.
                      \]
                      If we have a \textbf{failure} before this, the entire process is a \textbf{failure} and we apply Equation \eqref{eq:Rn-failure}.
		Denote by $\beta_{\epsilon}$ the probability of overcoming a barrier.
	\end{itemize}
A step is illustrated on Figure \ref{fig:RightBorder}, while Figure \ref{fig:MAC-R} illustrates a trajectory of $(R_n,K_n)$.

\begin{figure}[h]
    \centering
    \includegraphics[scale=0.5]{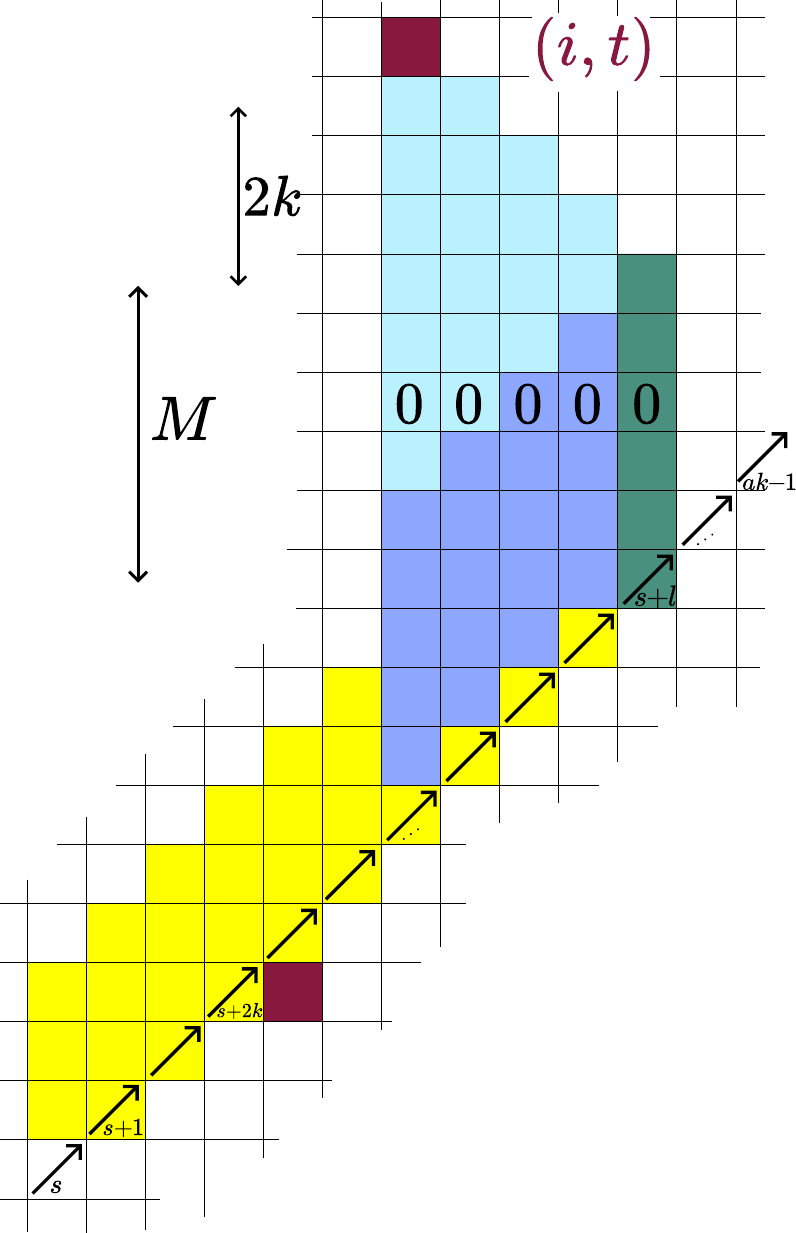}
    \caption{Suppose that $R_n=i$ and $t^\prime_n=t$. In green the column where we search for an arrow: here after $M-1$ ``other" cases (no errors in the $2k$ cells in diagonal, in dark blue) and no errors in the light blue area. As the arrow is of age $l$ and there is no errors on its path (yellow cells), it's a success. Then $R_{n+1} = i+4k+1-l$ and $t^\prime_{n+1}=t-M-l$.}
    \label{fig:RightBorder}
\end{figure}

\begin{figure}[!h]
\centering{}\includegraphics[scale=0.5]{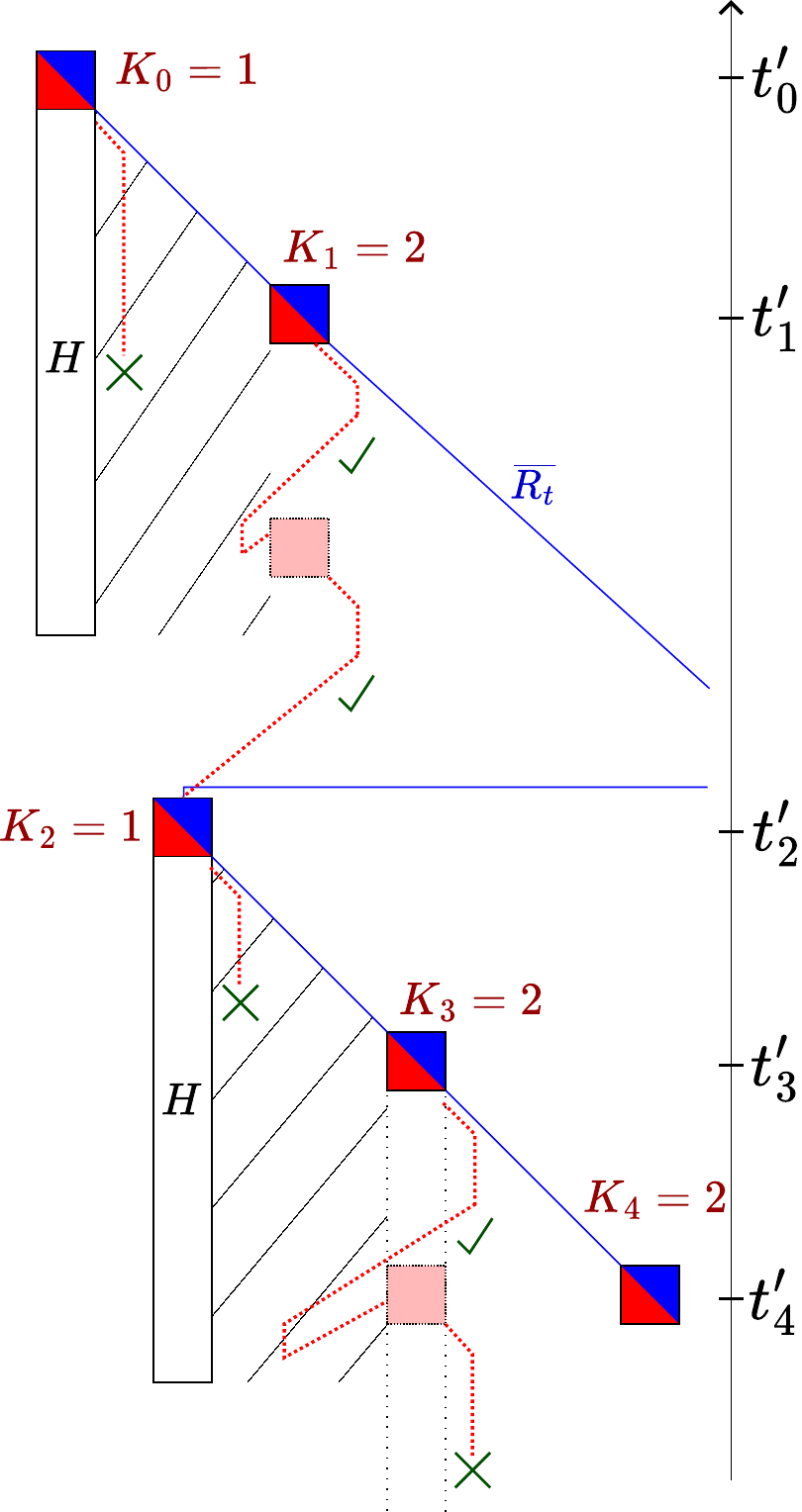} \caption{The first steps of $(R_{n},K_{n})$ and $\overline{R_{t}}$ (in blue, defined analogously to $\overline{L_{t}}$). The red-and-blue cells are at position $R_n$ at time $t^\prime_n$. The initial state is $(R_{0},K_{0})=(0,1)$. The first steps are: \protect \\
(1) The initial cell is not in a good synchronized zone: we encounter a bad case before $H$ checks, it's a failure. A barrier is represented and the new state is $(R_{1},2)$ with $R_{1}>R_{0}$.\protect \\
(2) The first check is a success (the cell is in a synchronized zone large enough), but before $H$ checks: it's a jump. The second check succeed after enough steps to overcome the barrier. The last success define the new level $R_{2}\protect\leq R_{1}$ and $K_{2}=1$.\protect \\
(3) It's a failure, $R_{3}>R_{2}$ and a barrier is represented.\protect \\
(4) The first check is a jump, bu the second one is a failure: $R_{4}<R_{3}$ and a new barrier is represented. \label{fig:MAC-R}}
\end{figure}

\subsubsection{Interpretation of $R_n$}

Analogously to the case of $L_n$, we use $\left(R_{n},t_{n}^{\prime}\right)$ to represent points where information must flow to its left to influence the left half-configuration at $t=0$. After a failure, the information can follow the natural slope of $1$ (Lemma \ref{lem:passageOfTime}). After a success, $R_n$ is on the border of a synchronized zone larger than $2k+1$ which is shifted at speed $1$ to the right, so the information cannot go through it.

We claim that the proof of Proposition \ref{prop:Ln-lt} still stands for the next proposition. Figure \ref{fig:Rn-rt} illustrates the left shift of the right border of the dependence cone after a success.
\begin{prop}
Let $\omega\in\Omega$. For all $n\in\N$, $r_{-t_{n}^{\prime}+2ak}\leq R_{n}+2ak$.
\end{prop}

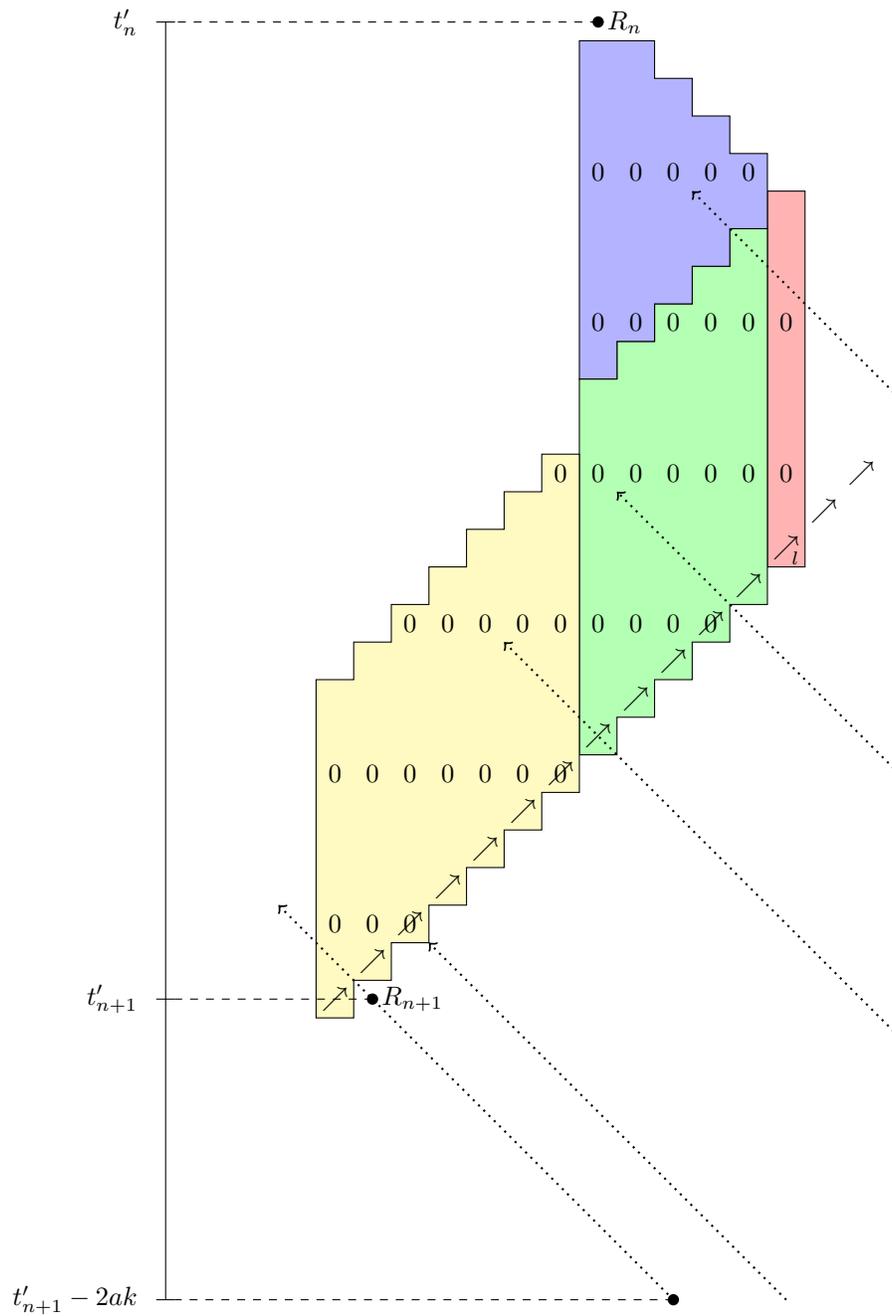
\begin{figure}[!h]
\centering{}
\begin{tikzpicture}[scale=0.5]

  \pgfmathsetmacro{\height}{10}

  \draw[fill=red!30] (5,5) rectangle ++(1,\height);
  
  \draw[fill=green!30] (0,0)
  \foreach \x in {0,...,3}{
    -| ++(1,1)
  }
  -| ++(1,\height)
  \foreach \x in {0,...,4}{
    -| ++(-1,-1)
  }
  -- cycle;

  \draw[fill=yellow!30] (0,0)
  \foreach \x in {0,...,6}{
    |- ++(-1,-1)
  }
  -- ++(0,9)
  \foreach \x in {0,...,5}{
    -| ++(1,1)
  }
  -| cycle;

  \draw[fill=blue!30] (0,\height)
  \foreach \x in {0,...,4}{
    -| ++(1,1)
  }
  \foreach \x in {0,...,3}{
    |- ++(-1,1)
  }
  -| cycle;

  \node [anchor=west] at (5.4,5.25) {${\scriptstyle l}$};
  \foreach \x in {-7,...,7}{
    \node at (\x+0.5,\x+0.5) {$\nearrow$};
  }
  
  \foreach \x in {-7,...,-5}{
    \node at (\x+0.5,-4.5) {$0$};
  }
  \foreach \x in {-7,...,-1}{
    \node at (\x+0.5,-0.5) {$0$};
  }
  \foreach \x in {-5,...,3}{
    \node at (\x+0.5,3.5) {$0$};
  }
  \foreach \x in {-1,...,5}{
    \node at (\x+0.5,7.5) {$0$};
  }
  \foreach \x in {0,...,5}{
    \node at (\x+0.5,11.5) {$0$};
  }
  \foreach \x in {0,...,4}{
    \node at (\x+0.5,15.5) {$0$};
  }

  \draw (-11,-14.5) -- (-11,19.5);
  \draw (-10.75,-14.5) -- ++(-0.5,0);
  \node [left] at (-11.5,-14.5) {$t^{\prime}_{n+1}-2ak$};
  \draw (-10.75,-6.5) -- ++(-0.5,0);
  \node [left] at (-11.5,-6.5) {$t^{\prime}_{n+1}$};
  \draw (-10.75,19.5) -- ++(-0.5,0);
  \node [left] at (-11.5,19.5) {$t^{\prime}_n$};

  \draw [thick, dotted, ->] (2.5,-14.5) -- ++(-10.5,10.5);
  \draw [thick, dotted, ->] (5.5,-14.5) -- ++(-9.5,9.5);
  \draw [thick, dotted, ->] (8.5,-7.5) -- ++(-10.5,10.5);
  \draw [thick, dotted, ->] (8.5,-0.5) -- ++(-7.5,7.5);
  \draw [thick, dotted, ->] (8.5,9.5) -- ++(-5.5,5.5);
  
  \fill (0.5,19.5) circle (0.15cm);
  \node [right] at (0.5,19.5) {$R_n$};
  \draw[dashed] (0.5,19.5) -- (-10.75,19.5);

  \fill (-5.5,-6.5) circle (0.15cm);
  \node [right] at (-5.5,-6.5) {$R_{n+1}$};
  \draw[dashed] (-5.5,-6.5) -- (-10.75,-6.5);
  
  \fill (2.5,-14.5) circle (0.15cm);
  \draw[dashed] (2.5,-14.5) -- (-10.75,-14.5);
  
\end{tikzpicture}
\caption{Illustration of a success in $R_{n}$ (not to scale). The interpretation is analogous to Figure \ref{fig:Ln-lt}. \label{fig:Rn-rt}}
\end{figure}

\begin{cor}
\label{cor:rtInfty}For a fixed $\omega\in\Omega$, if $R_{n}\tendsto n{+\infty}-\infty$ then $r_{t}\tendsto t{+\infty}-\infty$. 
\end{cor}

\subsection{Mean drift of the right border}
\label{subsec:Derive-Bordure-droite}

We can re-use the computations made for the left border with our new parameters:
\[
\alpha_{\epsilon}=\left(1-\epsilon\right)^{2k\left(2k+1\right)}\sum_{m=1}^{H}\left(p_{other}^{\epsilon}\right)^{m-1}p_{good}^{\epsilon}\tendsto{\epsilon}0 D_{a,k} \left(1-\frac{1}{a}\right)\eqqcolon\alpha
\]
and $\alpha\underset{a\to\infty}{=}1-\frac{5+\frac{2}{k}}{a}+o\left(\frac{1}{a}\right)$.

For the jump computations: if $N$ is the size of a jump,

\[
P\left(N=d\mid\text{\textbf{success}}\right)=\frac{\left(1-\epsilon\right)^{2k\left(2k+1\right)}(p_{other}^{\epsilon})^{d}\,p_{good}^{\epsilon}}{\left(1-\epsilon\right)^{2k\left(2k+1\right)}\sum_{m=1}^{H}(p_{other}^{\epsilon})^{m-1}\,p_{good}^{\epsilon}}
\]
and the same computations give
\[
P\left(N\geq\frac{H}{i}\mid\text{\textbf{success}}\right)\tendsto{\epsilon}0\left(\frac{1}{a}\right)^{1/i}\frac{1-\left(\frac{1}{a}\right)^{\frac{i-1}{i}}}{1-\frac{1}{a}}.
\]
The conclusion is then the same: for a fixed $m\in\N^{*}$, one can find a function $h:[0,1]\times\N\to[0,1]$ such that $\lim_{\epsilon\to0}h\left(\epsilon,a\right)\underset{a\to\infty}{\sim}\frac{m}{a}$ and $\beta_{\epsilon}\geq h\left(\epsilon,a\right)$. The mean drift of $\left(R_{n},K_{n}\right)$ is
\begin{align*}
\kappa^{\prime} & =\pi_{1}\sum_{m\in\Z}mP\left(R_{1}=m\mid K_{1}=1\right)+\pi_{2}\sum_{m\in\Z}mP\left(R_{1}=m\mid K_{1}=2\right)\\
 & =\pi_{1}\sum_{l=2k+1}^{ak-1}\left(4k+1-l\right)P\left(R_{1}=4k+1-l\mid K_{1}=1\right)+\pi_{2}4ak
\end{align*}
with $\pi=(\pi_{1},\pi_{2})\coloneqq\frac{1}{1-\alpha_{\epsilon}+\beta_{\epsilon}}(\beta_{\epsilon},1-\alpha_{\epsilon})$. The mean drift is then strictly negative if 

\begin{equation}
\beta_{\epsilon}\sum_{l=2k+1}^{ak-1}\left(l-4k-1\right)P\left(R_{1}=4k+1-l\mid K_{1}=1\right)>\left(1-\alpha_{\epsilon}\right)4ak.\label{eq:driftRightBorder}
\end{equation}
Observe that
\[
P\left(R_{1}=4k+1-l\mid K_{1}=1\right)=\frac{P\left(F_{0}^{0}\left(l\right)\right)}{P\left(\bigsqcup_{j=2k+1}^{ak-1}F_{0}^{0}\left(j\right)\right)}\tendsto{\epsilon}0\frac{2\left(ak-l\right)}{\left(\left(a-2\right)k-1\right)\left(a-2\right)k}
\]
and thus
\begin{align*}
\lim_{\epsilon\to0}\sum_{l=2k+1}^{ak-1}\left(l-4k-1\right)P\left(R_{1}=4k+1-l\mid K_{1}=1\right) & =\frac{2\sum_{l=2k+1}^{ak-1}\left(l-4k-1\right)\left(ak-l\right)}{\left(\left(a-2\right)k-1\right)\left(a-2\right)k}\\
 & \underset{a\to\infty}{\sim}\frac{ak}{3}.
\end{align*}
Finally, a sufficient condition is
\begin{equation}
\liminf_{\epsilon\to0}\beta_{\epsilon}>\frac{\left(1-\alpha\right)4ak}{\lim_{\epsilon\to0}\sum_{l=2k+1}^{ak-1}\left(l-4k-1\right)P\left(R_{1}=4k+1-l\mid K_{1}=1\right)}.\label{eq:conclusionDerive2}
\end{equation}
As $a\to\infty$, the right hand side is equivalent to $\frac{60+\frac{24}{k}}{a}<\frac{61}{a}$ (for $k>24$). The condition is then verified by choosing $m=61$: for $a$ large enough, $\kappa^{\prime}<0$ for $\epsilon$ small enough, thus $R_{n}\tendsto n{\infty}-\infty$ almost surely.

\subsection{Combining the left and right borders}

\begin{prop}\label{prop:Croisement p.s.}
If $L_n \tendsto{n}{\+\infty}+\infty$ and $R_n\tendsto{n}{+\infty}-\infty$ a.s., then $p_t(T_\epsilon)\tendsto{t}{+\infty}1$.
\end{prop}
\begin{proof}
Fix $\omega\in\Omega$ such that $L_n \tendsto{n}{\+\infty}+\infty$ and $R_n\tendsto{n}{+\infty}-\infty$. This implies that $l_t \tendsto{t}{+\infty}+\infty$ and $r_t \tendsto{t}{+\infty}-\infty$. So there exists a $T(\omega) < \infty$ such that $l_T > r_T$.

Then almost surely, there exists $T < \infty$ such that $x\mapsto \Psi^{T} \left(x ; U(\omega)\right)_0$ is constant by Lemma \ref{lem:emptyDependenceCone}. Therefore, $P\left(x\mapsto \Psi^t\left(x ; U(\omega)\right)_0 \text{ is constant}\right) \tendsto{t}{+\infty}1$.
\end{proof}

We can then finish the proof of Theorem \ref{thm:TransitionDePhase}. First fix a $k$ large enough so that it verifies the conditions stated in Section \ref{sec:errorsGLayer}. Then, fix $a$ large enough so that Equations \ref{eq:conclusionDerive1} and \ref{eq:conclusionDerive2} are verified. This means that for $\epsilon$ small enough, Equations \ref{eq:driftLeftBorder} and \ref{eq:driftRightBorder} are verified, and thus $L_{n}\tendsto n{+\infty}+\infty$ and $R_{n}\tendsto n{+\infty}-\infty$ almost surely. In other words, there exists an error rate $\epsilon_{1}$ such that for all $\epsilon\leq\epsilon_1$, the hypotheses of Proposition \ref{prop:Croisement p.s.} are verified and $T_{\epsilon}$ is ergodic. With our choice of $k$, $T_{\epsilon_2}$ is not ergodic.

\section{Conclusions}

This work shows that the ergodicity of a perturbed cellular automata with positive rate can have more than one phase transition depending on the value of the noise. An interesting direction to continue this work would be to understand which sets can be obtained as 
\[\{\epsilon\in[0,1]:\textrm{$F_\epsilon$ is ergodic}\}\]
where $F_\epsilon$ is a perturbation of a cellular automaton $F$ by a uniform noise of rate $\epsilon$.
Our construction is hardly adaptable, in part because we do not understand fully the invariant measures of the perturbation of the G\'acs cellular automaton.
In particular, we do not know how many phase transitions for ergodicity the perturbation of G\'acs cellular automaton realizes, nor do we know that our CA $T$ only realizes two.

\printbibliography{}

\end{document}